\documentclass[conference]{IEEEtran}
\IEEEoverridecommandlockouts
\usepackage{cite}
\usepackage{amsmath,amssymb,amsfonts,amsthm}

\newtheorem{theorem}{Theorem}

\usepackage[ruled,vlined]{algorithm2e}
\usepackage{graphicx}
\usepackage{textcomp}
\usepackage{float}
\usepackage{tabulary}
\usepackage{xcolor}
\usepackage{multirow}
\usepackage{subcaption}
\usepackage{url}

\def\BibTeX{{\rm B\kern-.05em{\sc i\kern-.025em b}\kern-.08em
    T\kern-.1667em\lower.7ex\hbox{E}\kern-.125emX}}
\begin{document}

\title{Towards Optimizing Storage Costs on the Cloud
}


\author{\IEEEauthorblockN{Koyel Mukherjee*, Raunak Shah*, Shiv Saini}
\IEEEauthorblockA{\textit{Adobe Research} \\
\{komukher, raushah, shsaini\}@adobe.com
}
\and
\IEEEauthorblockN{Karanpreet Singh\textsuperscript{\textdagger}, Khushi\textsuperscript{\textdagger},\\Harsh Kesarwani\textsuperscript{\textdagger}, Kavya Barnwal\textsuperscript{\textdagger}}
\IEEEauthorblockA{\textit{IIT Roorkee}\\
ksingh2@cs.iitr.ac.in, khushi@ee.iitr.ac.in, \\ hkesarwani@ee.iitr.ac.in, kbarnwal@cs.iitr.ac.in}
\and
\IEEEauthorblockN{Ayush Chauhan\textsuperscript{\textdaggerdbl}}
\IEEEauthorblockA{\textit{UT Austin} \\
ayushchauhan@utexas.edu}
}



\maketitle
\begingroup\renewcommand\thefootnote{*}
\footnotetext{These authors contributed equally.}
\endgroup
\begingroup\renewcommand\thefootnote{\textsuperscript{\textdagger}}
\footnotetext{Work done while interning at Adobe Research, India.}
\endgroup
\begingroup\renewcommand\thefootnote{\textsuperscript{\textdaggerdbl}}
\footnotetext{Work done while employed with Adobe Research, India.}
\endgroup

\begin{abstract}
  We study the problem of optimizing data storage and access costs on the cloud 
  while ensuring that the desired performance or latency is unaffected. We first propose an optimizer that optimizes the data placement tier (on the cloud) and the choice of compression schemes to apply, for given data partitions with temporal access predictions. Secondly, we propose a model to learn the compression performance of multiple algorithms across data partitions in different formats to generate compression performance predictions on the fly, as inputs to the optimizer. Thirdly, we propose to approach the data partitioning problem fundamentally differently than the current default in most data lakes where partitioning is in the form of ingestion batches. We propose access pattern aware data partitioning and formulate an optimization problem that optimizes the size and reading costs of partitions subject to access patterns.
  
  We study the various optimization problems theoretically as well as empirically, and provide theoretical bounds as well as hardness results. We propose a unified pipeline of cost minimization, called SCOPe that combines the different modules. We extensively compare the performance of our methods with related baselines from the literature 
  on TPC-H data as well as enterprise datasets (ranging from GB to PB in volume) and show that 
  SCOPe substantially improves over the baselines. We show significant cost savings compared to platform baselines, of the order of 50\% to 83\% on enterprise Data Lake datasets that range from terabytes to petabytes in volume. 
\end{abstract}


\section{Introduction}
\label{sec:intro}
Customer activities on the internet generate a huge amount of data daily. This is usually stored in cloud platforms such as Azure, AWS, Google Cloud etc. and is used for various purposes like analytics, insight generation and model training. With the massive growth in data volumes and usages, the costs of storing and accessing data have spiraled to new heights, significantly increasing the COGS (cost of goods sold) for enterprises, thus making big data potentially less profitable.

Cloud storage providers offer a tiered form of storage that has different costs and different throughput and latency limits across tiers. Table \ref{tab:azure} shows the storage costs, read costs and latency (measured as the time to first byte) for the different tiers of storage offered by Azure, a popular cloud storage provider. There is a clear trade-off between storage cost, read cost and latency across the tiers.

\begin{table}[htbp]
\caption{Cost and latency numbers for Azure\cite{azure}.}
\label{tab:azure}
\centering
 \begin{tabular}{|l|l|l|l|l|}
\hline
 & Premium & Hot & Cool & Archive \\ \hline
Storage cost cents/GB & \multirow{2}{*}{15} & \multirow{2}{*}{2.08} & \multirow{2}{*}{1.52} & \multirow{2}{*}{0.099} \\
(first 50 TB) &  &  &  &  \\ \hline
Read cost (cents, every 4 & \multirow{2}{*}{0.182} & \multirow{2}{*}{0.52} & \multirow{2}{*}{1.3} & \multirow{2}{*}{650} \\
MB per 10k operations) &  &  &  &  \\ \hline
Time to first byte & Single & ms & ms & \multirow{2}{*}{Hours} \\
 & digit ms &  &  &  \\ \hline
\end{tabular}
\end{table}

To further complicate things, enterprise workloads often exhibit non-trivial and differing patterns of access. Data access patterns are often highly skewed; only a few datasets are heavily accessed and most datasets see very few or $0$ accesses. Another common trend is recency, i.e., access frequency falls with age of dataset or file (Fig \ref{fig:data-access-stats2}). There are other interesting patterns too as shown in Fig \ref{fig:data-accesses-revised}. While some data see a constant number of read or write accesses, others may see periodic peaks, or read accesses decreasing over time. 
For marketing use cases, data is ingested for activation, leading to one-time read and write spikes followed by long inactive periods.
\begin{figure}[htbp]
\begin{minipage}[b]{0.49\linewidth}
\centering
  \includegraphics[width=\linewidth]{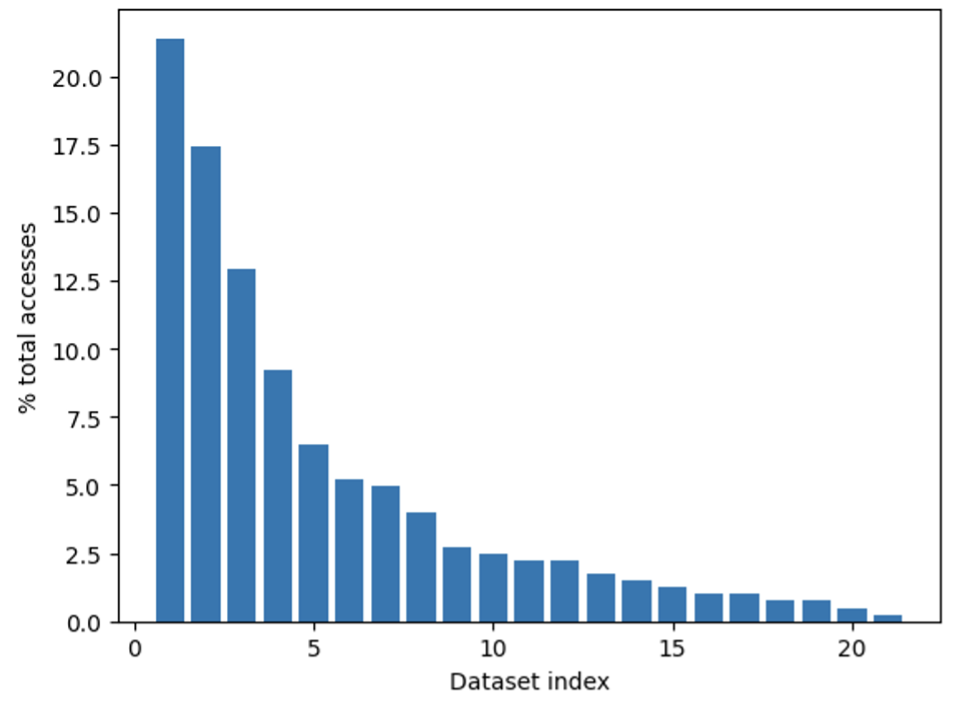}
  \subcaption{\small{\% accesses vs dataset index}}
\end{minipage}
\hfill
\begin{minipage}[b]{0.49\linewidth}
\centering
  \includegraphics[width=\linewidth]{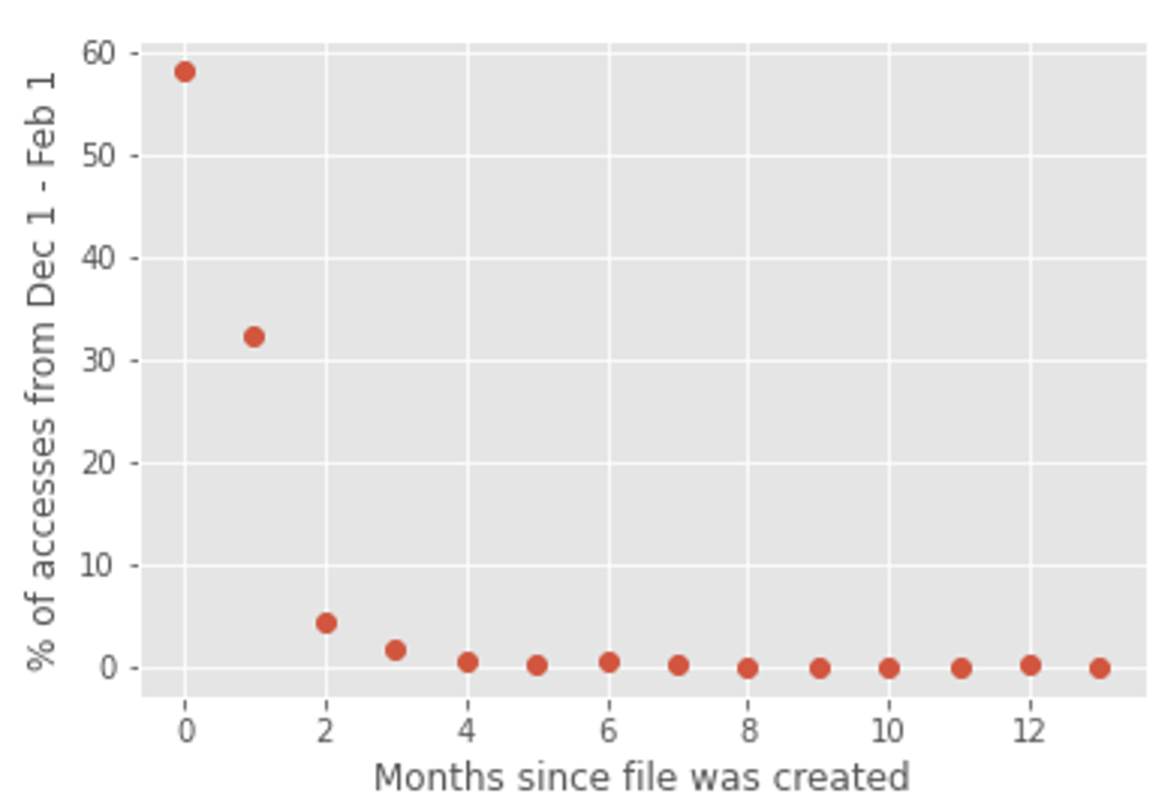}
  \subcaption{\small{\% accesses vs months since file was created}}
\end{minipage}
\caption{Enterprise Data access patterns}
\label{fig:data-access-stats2}
\end{figure}
Rule based methods such as pushing data to cooler tiers after $x$ days of inactivity fail due to seasonality and periodicity in access patterns, such as year-on-year analysis. Other intuitive rules such as caching the most recently accessed data in hot tiers are also ineffective, because even with a consistent number of accesses, it is non-trivial to determine the right tier for the data\footnote{Some cloud storage providers have recently started providing tiering and lifecycle management options based on last access \cite{aws-tiering, azure-lifecycle}. However, such methods are oblivious to varied and long term data usage patterns such as seasonal trends, year on year analysis, as well as required SLAs on certain types of data even if they are accessed less frequently.}.
Caching rules generally consider access related information (recency of access or frequency of accesses) for goals different from storage cost optimization. 
The optimal storage tier would depend on a trade-off between the storage costs (as determined by the size of the dataset and storage cost per unit size in each tier), access costs (determined by the amount of data being accessed, the access frequency, cloud provider's costs for different access types per unit size of data), the tier change costs, as well as the SLAs (i.e., latency and availability agreements with the clients). 

\begin{figure}[htbp]
\centering
  \includegraphics[width=\linewidth]{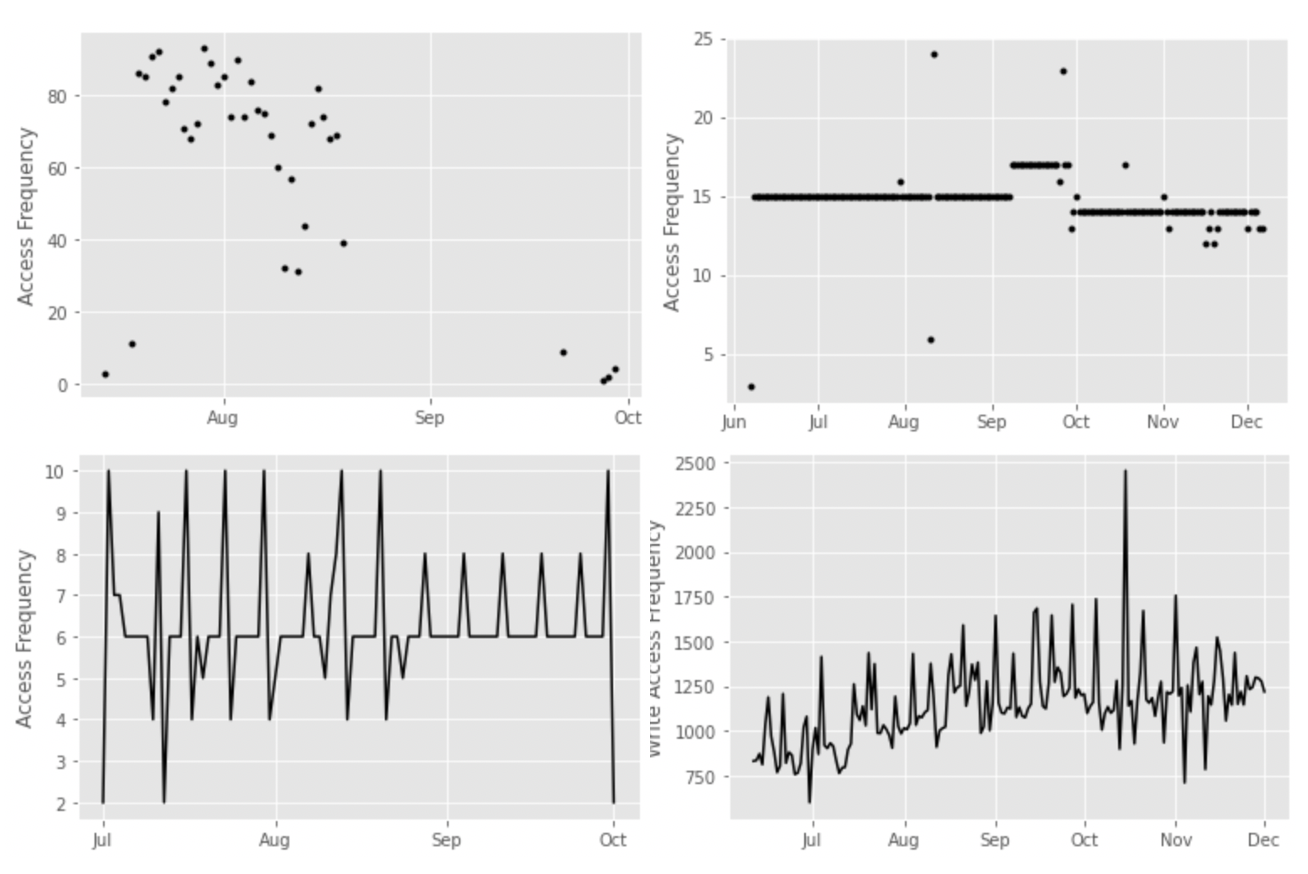}
\caption{Enterprise workloads on our data lake. \textbf{Top-left:} Read accesses decreasing over time for a particular dataset. \textbf{Top-right:} Read accesses remaining relatively constant over time for a particular dataset. \textbf{Bottom-left:} Periodic trend of read accesses for a certain class of datasets. \textbf{Bottom-right:} Write access trend across all datasets}
\label{fig:data-accesses-revised}
\end{figure}
Datasets in enterprise setting are often large, of the order of terabytes. Often, certain parts (files) of the data are more ``important'' or heavily accessed than the others, hence applying lifecycle management on entire datasets as a whole might be inefficient or sub-optimal.  Recently, workload aware and adaptive data partitioning schemes have been explored in the literature \cite{brendle2022sahara, sun2014fine, lu2017adaptdb, shanbhag2017robust, liwen_columnar_partitioning, mto_kraska}. However, these approaches often suffer from scalability issues as dataset sizes increase or query workloads becomes richer, and can have difficulty 
in dynamically adapting to changing workloads.

Another approach used to reduce the storage and read costs is data compression \cite{ares, devarajan2020hcompress}, but this can potentially add the overhead of decompression, which might increase the latency and compute cost. There have been efforts that enable efficient queries and analytics directly on compressed and/or sampled data, discussed further in Section II. A solution is thus needed which reduces the total storage, read, write and compute costs (which includes the necessary decompression costs, where unavoidable) from the cloud while meeting any service level agreements (SLAs).

\textbf{Our contributions:}
\begin{enumerate}
    \item We analyze the cost optimization problem theoretically and show it is strongly \textsc{NP-hard} (proof sketch). For special 
    cases, we give optimal, polynomial algorithms including an efficient greedy algorithm. 
    
    \item Our greedy algorithm is both scalable and effective. We apply it on enterprise Data Lake datasets of the order of \textbf{petabytes} in volume using real (historical) enterprise workloads that have substantial skew and other patterns. We show significant cost benefits, ranging from \textbf{~50\% to ~83\%} compared to platform baselines. The prediction model is near optimal with F1 $>0.96$ and does significantly better than intuitive baselines.
    
     \item We propose a Compression Predictor that predicts the compression ratio and decompression times for popular compression schemes on data partitions in different storage formats (csv, parquet) with good accuracy. We empirically study different features, models, data layouts, and data distributions across multiple schemes. 
     
    \item We study the access-pattern-aware data partitioning problem and show that it is strongly \textsc{NP-hard} (proof sketch). We propose a heuristic that achieves a good trade-off of space and cost empirically. For time series data, we give polynomial time approximation schemes. 
    
    \item We propose a unified pipeline of the above optimizers, predictor, and partitioner called SCOPe: Storage Cost Optimizer with Performance Guarantees, that allows tunable, access pattern aware storage 
    and access cost optimization on the cloud while maintaining SLAs. We provide substantial empirical validation and ablation studies on enterprise and TPC-H data. SCOPe outperforms related baselines by a significant margin. We also show that applying our partitioning heuristic can directly improve the baselines.
\end{enumerate}
The rest of the paper is organized as follows. We discuss the related work in Section \ref{sec:related}. We discuss the problem setting and datasets in Section \ref{sec:problem}. We study the cost optimization in Section \ref{sec:optimize}. We explain the compression predictor in Section \ref{sec:compression}. We study the query aware data partitioning problem in details in Section \ref{sec:partition}. We discuss the entire pipeline SCOPe (in comparison with baselines) in Section \ref{sec:experiments}. Finally, we conclude in Section \ref{sec:conclusion}.

\section{Related Work}
\label{sec:related}
Different aspects of the cloud storage problem have been studied in the literature. 

\textbf{Multi-tiering}
Recently, there has been some work that exploits multi-tiering to optimize performance, e.g., \cite{hermes, devarajan2020hcompress, devarajan2020hfetch, cheng2015cast} and/or costs, e.g., \cite{cheng2015cast, ERRADI2020110457, liu2019transfer, liu2021keep, si2022cost, kinoshita2021cost}. Storage and data placement in a workload aware manner, e.g., \cite{anwar2015taming, anwar2016mos, cheng2015cast} and in a device aware manner, e.g., \cite{vogel2020mosaic, lasch2022cost, lasch2021workload} have also been explored. \cite{devarajan2020hcompress} combine compression and multi-tiering for optimizing latency, but do not consider the storage, read and compute costs.   
Many of the existing works towards optimizing storage and read costs in a multi-tiered setting, e.g.,\cite{cheng2015cast, mansouricostoptim, mansouridatacenter, ERRADI2020110457, liu2019transfer, liu2021keep, si2022cost, kinoshita2021cost} propose policies for data transfer between tiers in online or offline settings, however there is no focus on data partitioning or data  compression. 
In general, we did not find a direct baseline for SCOPe that tries to optimize storage costs while maintaining latency guarantees, estimating compression performance and considering query aware data partitioning. However, we have modified some of the existing works on tiering as baselines for SCOPe, and we have extensively evaluated SCOPe against such baselines.
 
\textbf{Data Compression: }
Data compression has been heavily studied in the literature. 
While some of the works have studied compression that enables efficient querying, others have studied the memory footprint and cost aspects. 
Bit map compression schemes that enable efficient querying, such as WAH, EWAH, PLWAH, CONCISE, Roaring, VAL-WAH etc. \cite{wu2002compressing, wu2006optimizing, EWAH, PLWAH, roaring1, Roaring, concise, VALWAH} have been extensively studied, however these are more efficient and popular for read-only datasets. Several researchers have studied enabling efficient queries and analytics directly on compressed and/or sampled data, e.g.,  \cite{abadi2006integrating, succinct, compressdb, blinkdb}. Others have considered the cost aspect, e.g., \cite{ares, devarajan2020hcompress}. Yet others have looked at the dynamic estimation of compression performance, e.g., \cite{adaptive, devarajan2020hcompress} based on data type, size, similarity and distribution. 
In our setting, features suggested in literature did not work well, hence we proposed new, effective features. We have comprehensively evaluated multiple compression schemes, data layouts (parquet, csv, sorted data etc.), different sources of data (TPC-H and enterprise) and different query workload distributions (uniform, skewed), and evaluated the effect of prediction errors on the overall optimization.

The benefits of caching and computation pushdown in a disaggregated storage setting have been explored \cite{yang2021flexpushdowndb, yu2020pushdowndb}. However not all cloud service providers support computation pushdown for commonly used data formats such as parquet, and even if they do, it generally comes at a higher price. Moreover, not all types of queries can be accelerated by directly performing on compressed data or pushing down to the storage layer. Nevertheless, our optimization module can support query accesses that directly run on compressed data (i.e. no decompression) as well.

\textbf{Data Partitioning: }
Data partitioning has been studied in the database 
literature as a means of efficient query execution. 
Attribute based partitioning have been heavily studied 
\cite{jaya, zhou2010incorporating, sun2014fine, mto_kraska, liwen_columnar_partitioning, shanbhag2017robust, lu2017adaptdb}. 
Workload and workflow aware partitioning \cite{quiver, brendle2022sahara, diesel } has also been explored for different types of workflows. 
In particular, partitioning with respect to query workload has been explored by \cite{sun2014fine, mto_kraska, liwen_columnar_partitioning}. However, \cite{sun2014fine, liwen_columnar_partitioning} require row level labeling and hence the required compute is not scalable to enterprise data lake scale with $10^{12}$ rows in datasets of the order of terabytes.  \cite{liwen_columnar_partitioning} requires modification of data (by appending tuple ids to rows) which is difficult on client data due to access restrictions. \cite{liwen_columnar_partitioning} proposed tuple reconstruction which becomes non-trivial to apply in our setting in tandem with tiering and compression (for both cost and latency estimation to maintain performance guarantees), due to variability of parameters across tiers and compression schemes.\cite{sun2014fine, liwen_columnar_partitioning, mto_kraska} do not easily adapt to dynamically changing data and workloads. We give a novel graph based modeling at file and query level which is scalable, and easily adaptable to dynamically changing query workloads. Existing work considers disjoint partitions that often do not benefit the median query in a rich query set, whereas we allow overlapping partitions.

\section{Problem Setting}
\label{sec:problem}
Our problem is motivated by the cloud storage cost and performance considerations in the Adobe Experience Platform Data Lake, that is home to huge volumes of customer datasets, often being time series or event logs.
The data resides in the cloud (e.g., Azure) and we get a break-up of costs at the end of a billing period, e.g., a month. The cost is incurred for storage over a period of time and per every access. We consider optimizing the \textbf{total costs}\footnote{We use the cost parameters of ADLS Gen2. Similar analysis and modeling can be done with AWS, GoogleCloud and other cloud providers' parameters.} at the end of a billing period, say, `k' months for optimizing the COGS of the organization. Our model is run at the beginning of every billing period, to generate storage recommendations for all the datasets in the data lake for that billing period. The reader might be curious as to whether it adversely affects the cost to change storage tiers in such a periodic (batch processing) manner, instead of doing this much more frequently, in an ad-hoc way, especially if the data has an extremely fast changing access pattern. Note that changing storage tiers too frequently has the following disadvantages: i) early deletion charges per tier: data once moved to a tier, needs to reside there for a minimum period before we move it, otherwise we incur a penalty; ii) tier change costs:  for every tier change there are read and write costs incurred, iii) tier change scheduler: this would need to run with increased frequency, hence incurring high compute costs, offsetting the tier change benefits. Moreover, if data compression and data partitioning are also involved, frequent changes would result in a huge amounts of additional compute cost of data processing, thus increasing the COGS.


We study the following problems.
\begin{enumerate}
    \item \textsc{OptAssign}: Given predicted volume of accesses for datasets for a projected period, determine the optimal (in expectation) assignment of tier and compression scheme, while maintaining latency and capacity requirements.
    \item \textsc{ComPredict}: Accurately estimate the compression ratios and decompression speeds for different compression schemes on various datasets.
    \item \textsc{DataPart}: Access pattern aware optimal and efficient partitioning of data.
\end{enumerate} 
Finally, we present the results for the unified pipeline, SCOPe: Storage Cost Optimizer with Performance Guarantees. 

\textbf{Datasets and Workloads (Access logs):}
\begin{enumerate}
    \item Enterprise Data: 
    (a) \textbf{Enterprise Data I:} enterprise Data Lake data (ADLS Gen 2), with hundreds of datasets ranging from \textbf{TB to PB} in size for several customers. Here we only have access to meta data and historical access logs at dataset level. 
    (b) \textbf{Enterprise data II:} 3 tables for which queries as well as timestamp information are available. These are about 1.5 GB in total size. Here we have full access to data, but not the access logs, hence we have generated queries based on a skewed power-law (Zipf-like) distribution. 
    
    \item We use four variants of TPC-H data. It consists of 8 different tables of varying sizes with 22 different types of complex queries. We generated 20 queries from each query template and used them for experiments on our unified pipeline as well as for the compression predictor module in each case.
    (a) TPC-H 1GB with uniformly generated data, 
    (b) TPC-H Skew generated with Zipfian skew (high skew factor of 3),
    (c) TPC-H-100GB with uniformly generated data, and 
    (d) TPC-H 1TB with uniformly generated data.
    
\end{enumerate}
We used PostgreSQL for TPCH data, and Apache Spark for enterprise data. Our enterprise data is partitioned and stored on ADLS Gen2 in parquet format.

\section{\textsc{OptAssign}: Optimizing Overall Costs}
\label{sec:optimize}
\textsc{OptAssign} determines optimal (in expectation) assignment of tier and compression schemes, given data partitions with predicted number of accesses for the projected period. It assumes that compression performance prediction is available as a look up for the given data partitions, or, in absence of that, only optimizes the tier assignments. \textsc{OptAssign} maintains latency requirements, and also handles capacity constraints, in case there are storage  reservations on tiers.

\subsection{Mathematical Formulation}
Let the number of storage tiers or layers be $L$. The storage cost of layer $\ell \in [L]$ is $C^s_\ell$, read cost is $C^r_\ell$, write cost is $C^w_\ell$ and the read latency is $B_\ell$ seconds per unit data. Let the reserved capacity (space) for storage be $S_\ell$ and the compute cost per second be $C^c$. Layer $0$ denotes the lowest latency layer and $L-1$ denotes the archival layer with the highest latency. Typically, $S_{L-1}$ is $\infty$. 
The tier change cost from tier $u$ to tier $v$ is $\Delta_{u,v}$, and this includes reading from layer $u$, writing to layer $v$ and any other charges. 

Let there be $N$ data partitions $\mathcal{P}$ and for each partition $P_i$, the span (or, size) is $Sp(P_i)$, and the projected number of accesses is $\rho(P_i)$. There is also a latency threshold $T(P_i)$ associated with each partition. Among these, $\mathcal{I}$ (let $I = |\mathcal{I}|)$ denotes existing ones, and the remaining are newly ingested in the current billing period. The current tier assignment for $P_i$ is $L(P_i)$, and for newly ingested ones, we denote $L(P_i) = -1$. Now, the write costs for new partitions to tier $\ell$ can be written as $\Delta_{-1, \ell}$. In other words, $C^w_\ell = \Delta_{-1, \ell}$.  All existing partitions have a predicted number of accesses based on past behaviour and dataset characteristics. In the case of newly ingested data, this is approximately estimated based on data quality considerations, query patterns on similar historical data, or client specific/domain knowledge. There are $K$ compression algorithms, where one option is `no compression', and $R^k_i$ denotes the predicted compression ratio of algorithm $k$ on $P_i$, and similarly $D^k_i$ denotes the predicted decompression time (for `no compression', $R^k_i$ is $1$ and $D^k_i$ is $0$ for all $i$). The compression scheme applied to a partition $P_i$ is $K(P_i)$. 

We give an ILP for the problem \textsc{OptAssign}. $x_{n,\ell, k}$ is an indicator variable that is $1$ when 
partition $P_n$ is assigned to tier $\ell$ with compression scheme $k$, and 
$0$ otherwise. 
$\alpha$, $\beta$, $\gamma$ are hyper-parameters to decide the weight for corresponding cost terms.
The first term in the objective function represents the cost of writing data (either new data or existing data from another tier) and then storing it in a tier after applying a particular compression algorithm. The second term represents the 
expected decompression cost 
(compute cost) and read cost of the merged partitions. 
The first 
equality is for feasibility purposes: every partition must go to 
one tier and at most one compression algorithm can 
be applied to it. The second inequality constraint ensures that the data stored does not exceed the capacity for that layer, where the capacity 
is determined by capacity reservations on the cloud. (Note that in case there are 
no reservations, hence no upper bound, the capacity would be $\infty$.) The third inequality is to ensure that decompression and read don't cause overhead in latency, and are less than the maximum latency threshold for any partition. 
Finally, the last equality forces that for existing partitions, the compression scheme does not change once applied; this is imposed to prevent additional latency and operational costs of frequently changing the compression of data partitions. Note that the above ILP can become infeasible due to capacity restrictions and latency constraints. In that case, the latency requirements need to relaxed iteratively till a feasible solution is found. However, in this case, there can be no solution satisfying all constraints, and it is a limitation of the system constraints, and not the solution. 
\begin{align}
\label{eq:ilp}
  \begin{split}
  \min \sum_{n=1}^{N}\sum_{k=1}^{K}\sum_{\ell=1}^{L}  [\left(\alpha\ C^s_\ell + \gamma \Delta_{L(P_n), \ell}\right)\  \frac{Sp(P_n)}{R^k_n}& \\
  +\ \beta \rho(P_n)\left( C^c\ D^k_n \ +\ C^r_\ell \frac{Sp(P_n)}{R^k_n}\right)]\ x_{n, \ell, k}&
  \end{split}\\\nonumber
  \text{s.t. } \sum_{\ell=1}^{L} \sum_{k=1}^{K} x_{n, \ell, k} \ = \ 1\ ,\forall\  n\ \in [N]\\\nonumber
  \sum_{n=1}^{N}\sum_{k=1}^{K}{\frac{Sp(P_n)}{R^k_n} x_{n, \ell, k}} \ \leq S_\ell \ ,\forall\  \ell\ \in [L]\\\nonumber
  \sum_{\ell=1}^{L} \sum_{k=1}^{K} {\left( D^k_n \ + B_\ell\right) x_{n,\ell, k}} \ \leq \ T(P_n)\ ,\ n\ \in \ [N]\\\nonumber
  x_{n,\ell, k} \in \{0, 1\} \ \forall n \in [N],\  \ell \in [L], \ k \in [K]\\\nonumber
  x_{n,\ell, k} = 0 \ \forall n \in [I],\  \ell \in [L], \forall k \neq K(P_n)
\end{align}
The ILP can be extended to handle the scenario of computation pushdown or compression schemes allowing certain operations directly on the compressed data. Let $f$ fraction of queries be amenable to such a pushdown and the remaining $(1-f)$ fraction would require decompression. 
Then only $(1-f) \rho(P_i)$ would contribute to the read and decompression-compute costs in the objective function for partition $P_i$, and similarly to the latency constraint, while the remaining fraction, $f \rho(P_i)$ would have $0$ contribution to either. 
In this way \textsc{OptAssign} can handle a partial storage disaggregation. Also note that \textsc{OptAssign} is a general framework and can easily handle following scenarios: a total storage capacity provisioned per tier by the enterprise\footnote{In this case, the $L$  constraints on per layer storage capacity $S_\ell$ would be replaced by a single constraint, on the total sum of the storage used across all layers being bound by a capacity $S$.}, a customer specific capacity per tier \footnote{Here, the per layer storage constraint would be replaced by $Q$ constraints, one per customer for a total of $Q$ customers.}, an unlimited (infinite) capacity per tier where there is no pre-determined storage entitlement but billing is per usage\footnote{Here, $S_\ell = \infty \ \forall L$, hence the L constraints on capacity can be removed.} These variations would be determined from customer licensing agreements, pricing by the cloud storage operator, and internal cost considerations and demand projections. Moreover, by tuning the weights in the objective, one can give more weight to one type of cost over others as required by the application. (We show this in Section \ref{sec:experiments}).

\begin{theorem}
\label{thm:opt-np-hard}
\textsc{OptAssign} is strongly \textsc{NP-hard}. 
\end{theorem}
\begin{proof}
This follows by a reduction from \textsc{3-Partition}.
Here we provide a proof sketch due to limited space. 
Consider an instance $\mathcal{I} =  \{a_i\}$ of \textsc{3-Partition} with $n=3v$ integers such that $\sum_{i\in [n]}{a_i} = B v$. The decision question is whether there exists 
a partition of $\mathcal{I}$ in to $v$ groups, such that each group 
sums to exactly $B$. Now, construct a relaxed instance of 
\textsc{OptAssign}, where cost parameters are $0$, $K=0$ (no compression algorithms) and latency thresholds are met by all tiers. We create a data partition of span $=a_i$ for each $a_i \in \mathcal{I}$. Let there be $v$ tiers of storage, with a capacity 
limit of $B$ per tier. It can be seen that a \textsc{Yes} instance in the \textsc{3-Partition} instance corresponds to a \textsc{Yes} instance in the \textsc{OptAssign} instance, and vice versa.
\end{proof}

\subsection{Polynomial Algorithms for Special Cases}
\subsubsection{Equal sized Partitions, No compression}
Consider the case where the data partitions are of equal spans (i.e., $Sp(P_i) = S \ \forall i\in [N]$ for some $S$),  
and there are no compression schemes ($K=0$)\footnote{We \textbf{do not require partitions to be equal sized} and can handle different compression algorithms.  The \textbf{datasets are of varying sizes in our experiments}. However, it is a possible scenario that can be enforced from the system administrator end, as a part of the data ingestion workflow, if the client requires this functionality or for facilitating data management. For example, one can configure file sizes to a certain set value by the command ` write.target-file-size-bytes’ as per https://iceberg.apache.org/docs/latest/configuration/), or specify the default parquet compression level to be null.}. Let all partitions be ingested at the same time, and no prior assignments exist. This is a possible scenario in practice when for a given storage account 
existing data are purged periodically, and new set of data are ingested. 
Since the span of each of the $N$ data partitions is $S$, the capacity of 
each storage tier can be expressed as a multiple of $S$, without loss of any generality. Also, as earlier, archive tier has capacity $\infty$. 
Since the data partitions are of equal sizes, we can consider the partitions to be of unit size. Now the capacity of 
each layer $\ell$ can be expressed as an integer $Z_\ell$, where $Z_\ell = \min\{N,\lfloor\frac{S_\ell}{S}\rfloor\}$, 
because there are at most $N$ partitions that need to be assigned. 

Now, let us construct a bipartite graph $\mathcal{G} = (\mathcal{U}, \mathcal{V}, \mathcal{E}, \mathcal{W})$. There would be $N$ nodes corresponding to the 
$N$ data partitions in $\mathcal{U}$. For each of the $L$ tiers (including archive), 
create $Z_\ell$ number of nodes 
in $\mathcal{V}$. Denote these as $Z_\ell$ copies of tier $\ell$.
The edge $e = (u,v)$ between a data segment and every copy of a tier 
would exist only if the latency threshold of the data segment would not be 
violated by assigning the segment to the tier.   
The weight of each such edge 
would be determined by the storage cost of the tier and the expected
read cost from that tier, as determined by the projected number of accesses of that data segment. 
Now, we solve a minimum weighted bipartite matching problem in this bipartite graph. 
Note that the size of the bipartite graph is polynomial since, there are at most 
$N + N\ L $ nodes. The minimum weighted matching would select the edges 
of minimum total weight, such that every node is assigned to at most one 
copy of at most one tier. There would be at most $Z_\ell$ assignments to any tier, 
since there are only $Z_\ell$ nodes corresponding to each tier that can be selected 
by the matching. The selected edges would not violate the latency thresholds, 
as the edges exist only if the threshold would be satisfied by the assignment. 
Therefore, the assignment found by minimum weight bipartite matching is not only feasible, in terms of latency requirement for assignments 
to regular tiers, but also optimal in terms of overall costs for the projected period. The edge weights can also 
be tuned based on the chosen hyperparameters and weightage of the various cost 
factors. The run time is polynomial in the input size: $O(N^2 L^2 E)$. 
Fig.\ref{fig:partitioning}(b) shows the above construction. 
\begin{theorem}
There exists a polynomial time optimal algorithm for the case of equal  sized data 
partitions and no compression. 
\end{theorem}
\begin{proof}
The proof follows from the above discussion. Specifically, the assignment found 
by the matching is feasible by construction, 
and the overall 
weight of the edges chosen is minimum by the optimality of the minimum 
weight bipartite matching algorithm. The time complexity follows from the that of Hungarian method.  
\end{proof}

\subsubsection{Unbounded Capacity}
Consider the general version of \textsc{OptAssign} (unequal sized partitions, multiple compression 
schemes) with the relaxation of the 
capacity constraints. Specifically, there are no capacity bounds on the tiers. This is a commonly occurring scenario in practice, including in our private enterprise Data Lake setting. 
In this case, a simple greedy algorithm gives the optimal solution. For every partition $P_i$, compute the set of feasible tuples $(\ell, k)$ of tier and compression algorithm,  and the choosing the lowest cost option per partition. This is feasible since there are no capacity restrictions and gives the optimal solution overall. The run time is $O(N\ L\ K)$, and for constant $L$ and $K$, this becomes linear in the number of partitions. 

\begin{theorem}
\label{thm:greedy}
There exists an optimal polynomial time algorithm for \textsc{OptAssign} 
when there are no capacity constraints. 
\end{theorem}
\begin{proof}
Since the greedy algorithm evaluates the lowest cost option for every merge, the 
overall cost is the lowest. If the overall cost is not the lowest, there has to be at least one feasible assignment 
of lower cost, but the greedy algorithm would have considered at one of the options, and would have therefore selected it. 
\end{proof}

\subsection{Empirical Validation on Enterprise Data}
We applied \textsc{OptAssign} with $K=0$ on Enterprise Data I using datasets as the data partitions, and projected access patterns for $6$ months using historical access logs. These datasets are 
large, ranging from TB to PB. We observed significant cost reduction benefits over the platform baseline. Figure \ref{fig:cost_benefit_6month} shows the \% cost benefit vs size and number of accesses for files for $6$ month projections for one customer account. We show in Table \ref{tab:customer} the projected cost benefits for 4 different customer accounts. Our methods are scalable and computationally efficient (e.g., the optimization took 2.53s on 463 datasets of customer B from Table \ref{tab:customer}).

\begin{table}[htbp]
\caption{\% cost benefits for data across 4 customers.}
\centering
\begin{tabular}{|l|c|c|c|}
\hline
& \multirow{2}{*}{\textbf{Total Size (PB)}} & \multicolumn{2}{c|}{\textbf{\% Cost Benefit}} \\
& & 2 mos & 6 mos \\ \hline
Customer A & 0.56 & 10.59 & \textbf{61.6} \\ \hline
Customer B & 0.45 & 8 & \textbf{53.72} \\ \hline
Customer C & 0.053 & 11.58 & \textbf{83.69} \\ \hline
Customer D & 0.085 & 9.93 & \textbf{49.6} \\ \hline
\end{tabular}
\label{tab:customer}
\end{table}

\begin{figure}[htbp]
\centering
\begin{minipage}{0.49\linewidth}
    \includegraphics[width=\linewidth]{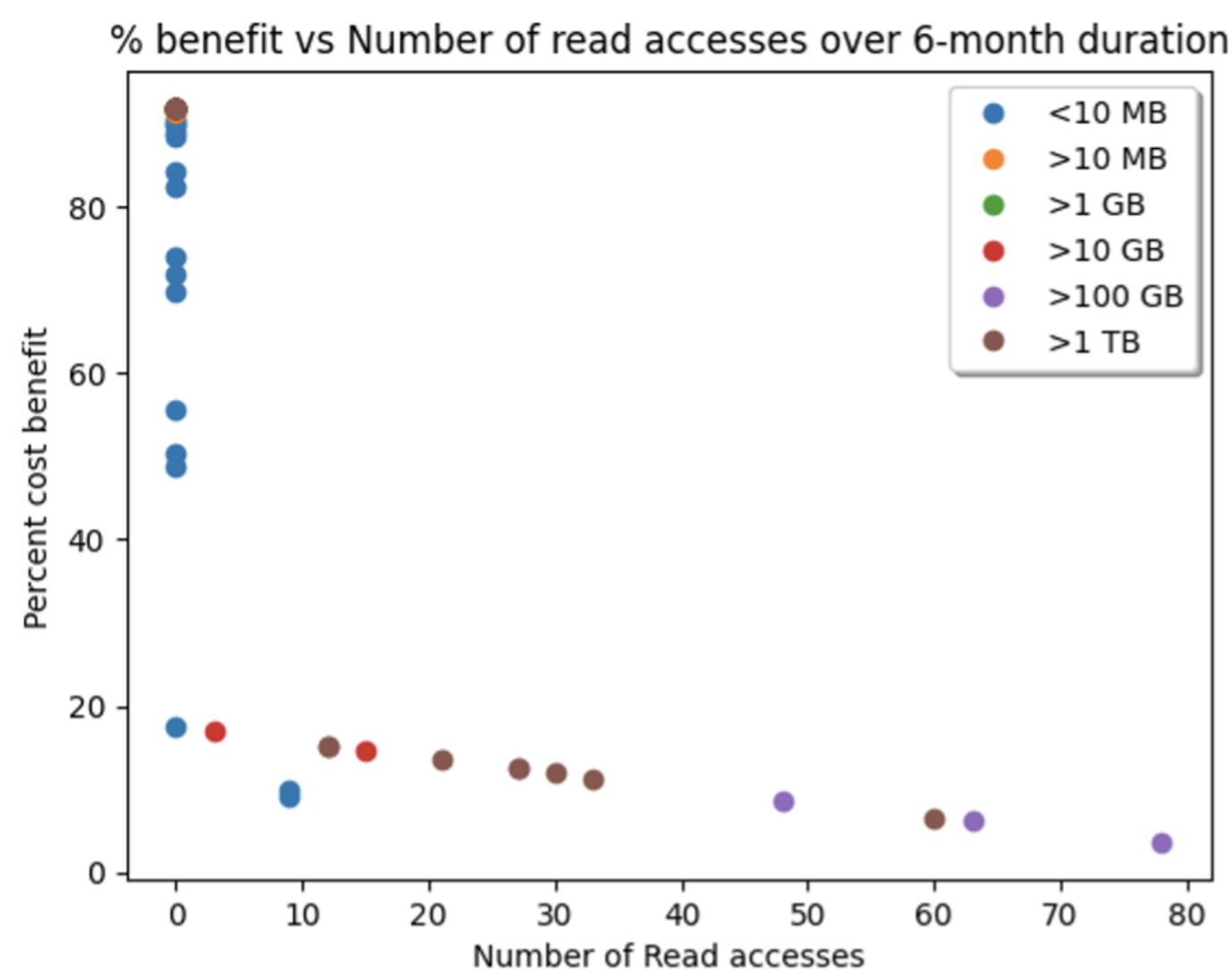}
    \subcaption{Cost benefit vs read accesses.}
\end{minipage}
\hfill
\begin{minipage}{0.49\linewidth}
    \includegraphics[width=\linewidth]{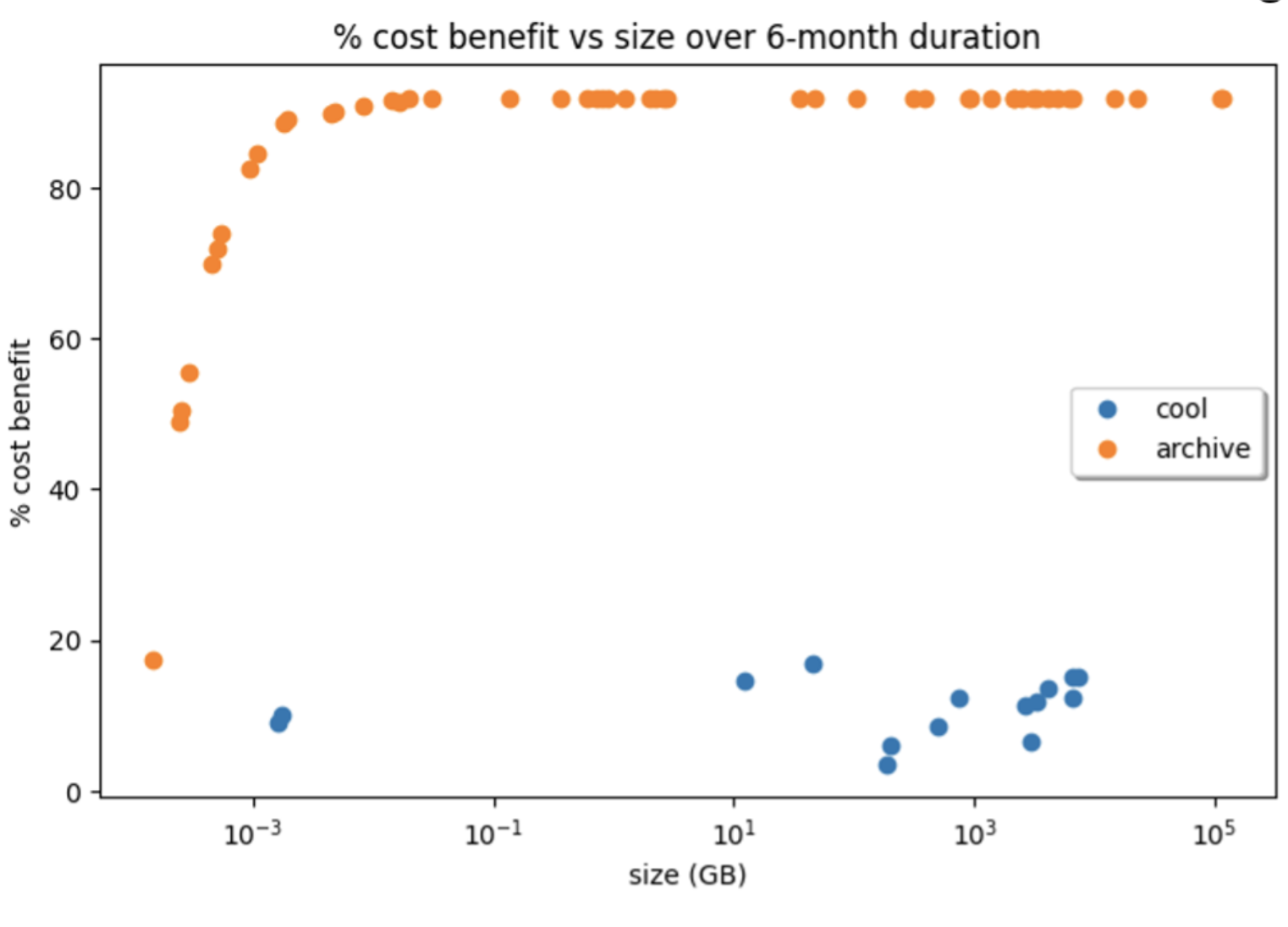}
    \subcaption{Cost benefit vs size.}
\end{minipage}    
\caption{Projected percent cost benefit for next 6 months (considering transfer from hot tier to cool+archive tiers)}
\label{fig:cost_benefit_6month}
\end{figure}

\textbf{Predicting Access Patterns and Quantifying Errors}
Predicting access patterns is a non-trivial problem. We have proposed a Random Forest model that is near optimal, with high precision and recall (F-1 score $>0.96$). Table \ref{tab:confusion_tier} shows the confusion matrix for one storage account. Random Forest showed the best overall performance compared to others like Gradient Boosted Trees and LSTMs. 
 \begin{table}[htbp]
\caption{Confusion matrices for predicted vs ideal tiering for one storage account (Around 700TBs of data in 760 datasets) over a 2 month prediction horizon.}
\label{tab:confusion_tier}
\centering
\begin{tabular}{cccc}
 &  & \multicolumn{2}{c}{\textbf{Ideal Tier}} \\ \cline{2-4} 
\multicolumn{1}{c|}{} & \multicolumn{1}{c|}{} & \multicolumn{1}{c|}{Hot} & \multicolumn{1}{c|}{Cool} \\ \cline{2-4} 
\multicolumn{1}{c|}{\multirow{2}{*}{\textbf{Predicted Tier}}} & \multicolumn{1}{c|}{Hot} & \multicolumn{1}{c|}{291} & \multicolumn{1}{c|}{12} \\ \cline{2-4} 
\multicolumn{1}{c|}{} & \multicolumn{1}{c|}{Cool} & \multicolumn{1}{c|}{12} & \multicolumn{1}{c|}{445} \\ \cline{2-4} 
\end{tabular}
\end{table}
We used \textsc{OptAssign} to assign the ground truth label encoding (i.e. the optimal tier) for each dataset while training the model. For prediction experiments, we ensured out-of-time validation and testing, and ran experiments on Apache Spark. The features used were (i) dataset size, (ii) months since dataset creation, and aggregated monthly (iii) read and (iv) write accesses for the last few months. The model needs to be retrained periodically for the next cycle (whose duration is adjustable) to account for changes in access pattern distributions. The cost for this batch job and the compute cost of applying the tiering operation is negligible compared to the tier change costs of the public cloud storage provider. The \% benefit shown is computed after deducting this cost, which shows that our model is practically feasible. Table \ref{tab:accesspatterns} shows how our results are consistent over multiple prediction horizons as well as how we compare to other tiering baselines, including caching based ones. Note that even after making errors, the \% benefit is close to the ideal case where all access information is known beforehand. Moreover, making significant mistakes is unlikely given the typical skewed or seasonal enterprise query workloads which are predictable enough by ML models to determine the optimal tier. Also note that the benefit is higher when we look at longer prediction horizons since lesser number of tier changes are required, as expected. Bringing in the archive layer helps increase the \% as well. Regarding performance considerations - while expected latencies are bounded as required by SLAs, there can be occasional unexpected accesses, causing tail latencies to be longer. A comparison with additional nontrivial baselines that also consider partitioning, latency, and compression along with multi-tiering is given at the end of the paper. 
Caching inspired baselines in rows 2 and 3 perform poorly because of two main reasons. i) Firstly, recency of access does not guarantee access in the next projected period, and one would need to train an ML model for that. ii) Secondly, even if access prediction is $100\%$ correct and a dataset is certainly going to be accessed in the next billing period, the optimal tier (from cost perspective, subject to performance considerations) might still not be hot, given the dataset size, the number and type of accesses, the amount of data to be accessed, tier change costs (if applicable), cloud cost parameters, as well as SLA and availability agreements with clients, which vanilla caching rules do not consider.
\begin{table}[htbp]
\caption{Comparison of OptAssign (with predicted or known access information) with intuitive baselines for the same storage account in Table \ref{tab:confusion_tier}.}
\label{tab:accesspatterns}
\resizebox{\linewidth}{!}{%
\begin{tabular}{|c|c|c|c|}
\hline
\multirow{2}{*}{\textbf{Model}} & \textbf{Access} & \textbf{Duration} & \multirow{2}{*}{\textbf{Benefit}} \\
 & \textbf{Information} & (\textbf{months}) &  \\ \hline
All hot & N/A & 2 & 0\% \\ \hline
``Hot" if data accessed in last 2 mos & N/A & 4 & 2.67\% \\ \hline
``Hot" if data accessed in last 1 mo & N/A & 4 & 3.25\% \\ \hline
Use optimal tier of prev. month & N/A & 2 & 5.07\% \\ \hline
OptAssign (Hot, Cool) & Predicted & 2 & 9.570\% \\ \hline
OptAssign (Hot, Cool) & Predicted & 4 & 13.58\% \\ \hline
OptAssign (Hot, Cool) & Known & 2 & 9.574\% \\ \hline
OptAssign (Hot, Cool) & Known & 4 & 13.62\% \\ \hline
OptAssign (Hot, Cool) & Known & 6 & 15.39\% \\ \hline
OptAssign (Hot, Cool, Archive) & Known & 6 & 43.8\% \\ \hline
\end{tabular}}
\end{table}

\section{\textsc{ComPredict}: Compression Predictor}
\label{sec:compression} 
We present \textsc{ComPredict}, which estimates compression ratios and decompression speeds for data partitions on the fly.
This involves training a model that is a one-time task, assuming the distribution of data types and other features remain largely unchanged\footnote{This has to be repeated at periodic intervals to handle slow changing data type distributions}. 
The model is trained to predict for a few popular compression schemes: gzip, snappy, and lz4 and two different data storage layouts (row-store and column-store). We found that our method also works well on other compression schemes like bz2, zlib, lzma, lzo, and quicklz, however we have omitted those results due to lack of space.


\textbf{Data Sources and Features:} Intuitively, compression performance can depend on various factors, such as, choice of compression scheme, data storage layout (row ordering vs column ordering), size of datasets, and characteristics of the data, e.g. data types, repetition in the data, entropy in the data,  organization of data contents (sorted vs unsorted) among others. 
Existing approaches for predicting comrpession performance  often use random samples from the dataset and simple features based on size or datatype. From Fig. \ref{fig:size-entropy} we can see that a sample formed from randomly sampled rows is typically not a good representation of the data that is usually queried from tabular datasets. We propose that this is because queried data typically has more repetition, which results in higher compression ratios compared to random samples. Note that if samples are generated in query aware pattern, skew in query workload can also have an effect. Considering only features like dataset size, datatype, or assuming a fixed data distribution like in prior art\cite{devarajan2020hcompress} is not enough to capture such notions. 
We created `weighted entropy' features for each partition $P$, with one feature for each data type present in $P$:\\
    $H(P, d) = -\sum_{s \in P[:,\;d]} len(s)\times pr(s)\times\log(pr(s))$, $d\in D$
    
Here $D$ denotes the set of datatypes of columns present in the partition $P$ (e.g. int, float, object, etc). For all strings $s$ that occur within the columns of a particular datatype $d$, we compute the probability of occurrence $pr(s)$ and length $len(s)$ of each string. $H(P, d)$ gives us an approximate representation of the amount of repetition in the table with datatype $d$. Computing these features requires a one-time full scan of each partition. 
\begin{figure}[htbp]
    \centering
    \includegraphics[width=\linewidth]{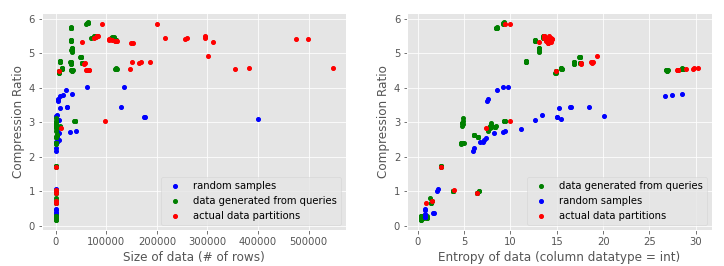}
    \caption{Compression Ratio vs Size (left) and Compression Ratio vs Entropy (right) on TPC-H dataset using gzip.}
    \label{fig:size-entropy}
\end{figure}
The samples used to train the model are derived from results of queries run on partitions. The number of samples required depends on the nature of the query workload.
Computing the samples and features took around 2-3 hours. Training the model takes a few seconds, and inference is almost instantaneous.
Table \ref{tab:comp-feature-comparison} shows a comparison of data samples (random vs query based) and features (size vs weighted entropy) for prediction on gzip, using Random Forest model. From Fig. \ref{fig:size-entropy} and Table \ref{tab:comp-feature-comparison} we can conclude that query based sampling using weighted entropy features are more effective for prediction. 
\begin{table}[htp]
\caption{Compression Ratio \& Decompression Speed Prediction by Random Forest model for Various Training Data and Features (GZIP Compression on TPC-H 1GB)}
\resizebox{\linewidth}{!}{%
\begin{tabular}{|c|c|c|c|c|c|}
\hline
 & \textbf{Training Data} & \textbf{Features} & \textbf{MAE} & \textbf{MAPE} & \textbf{R2} \\ \hline
\multirow{3}{*}{\textbf{Compression Ratio}} & Random Samples & Weighted Entropy & 1.022 & 72.188 & -0.656 \\ \cline{2-6} 
 & Queries & Size & 0.049 & 3.013 & 0.995 \\ \cline{2-6} 
 & Queries & Weighted Entropy & \textbf{0.021} & \textbf{0.527} & \textbf{0.988} \\ \hline\hline
\multirow{3}{*}{\textbf{Decompression Speed}} & Random Samples & Weighted Entropy & 18.713 & 268.627 & 0.069 \\ \cline{2-6} 
 & Queries & Size & 2.398 & 5.555 & 0.792 \\ \cline{2-6} 
 & Queries & Weighted Entropy & \textbf{0.254} & \textbf{1.215} & \textbf{0.989} \\ \hline
\end{tabular}}
\label{tab:comp-feature-comparison}
\end{table}

\begin{table*}[htp]
\centering
\caption{Compression Ratio Prediction for Various Models, Compression Schemes, \& Data Layouts (TPC-H 1GB)}
\resizebox{0.9\textwidth}{!}{%
\begin{tabular}{|c|ccc|ccc|ccc|ccc|ccc|}
\hline
\multirow{2}{*}{\textbf{Model}} & \multicolumn{3}{c|}{\textbf{gzip}} & \multicolumn{3}{c|}{\textbf{snappy}} & \multicolumn{3}{c|}{\textbf{parquet + gzip}} & \multicolumn{3}{c|}{\textbf{parquet + snappy}} & \multicolumn{3}{c|}{\textbf{parquet + lz4}} \\
 & \textbf{MAE} & \textbf{MAPE} & \textbf{R2} & \textbf{MAE} & \textbf{MAPE} & \textbf{R2} & \textbf{MAE} & \textbf{MAPE} & \textbf{R2} & \textbf{MAE} & \textbf{MAPE} & \textbf{R2} & \textbf{MAE} & \textbf{MAPE} & \textbf{R2} \\ \hline
\textbf{Averaging} & \multicolumn{1}{c|}{0.215} & \multicolumn{1}{c|}{5.353} & - & \multicolumn{1}{c|}{0.074} & \multicolumn{1}{c|}{3.315} & - & \multicolumn{1}{c|}{0.781} & \multicolumn{1}{c|}{23.154} & - & \multicolumn{1}{c|}{0.531} & \multicolumn{1}{c|}{20.101} & - & \multicolumn{1}{c|}{0.483} & \multicolumn{1}{c|}{19.494} & - \\ \hline
\textbf{XGBoost} & \multicolumn{1}{c|}{0.033} & \multicolumn{1}{c|}{0.851} & 0.991 & \multicolumn{1}{c|}{0.017} & \multicolumn{1}{c|}{0.733} & 0.991 & \multicolumn{1}{c|}{0.057} & \multicolumn{1}{c|}{1.482} & 0.989 & \multicolumn{1}{c|}{0.040} & \multicolumn{1}{c|}{1.305} & 0.988 & \multicolumn{1}{c|}{0.036} & \multicolumn{1}{c|}{1.206} & 0.992 \\ \hline
\textbf{Neural Network} & \multicolumn{1}{c|}{0.030} & \multicolumn{1}{c|}{0.793} & 0.993 & \multicolumn{1}{c|}{0.02} & \multicolumn{1}{c|}{0.930} & 0.985 & \multicolumn{1}{c|}{0.062} & \multicolumn{1}{c|}{1.549} & 0.991 & \multicolumn{1}{c|}{0.049} & \multicolumn{1}{c|}{1.730} & 0.992 & \multicolumn{1}{l|}{0.047} & \multicolumn{1}{l|}{1.747} & \multicolumn{1}{l|}{0.990} \\ \hline
\textbf{SVR} & \multicolumn{1}{c|}{0.071} & \multicolumn{1}{c|}{1.920} & 0.977 & \multicolumn{1}{c|}{0.069} & \multicolumn{1}{c|}{3.049} & 0.885 & \multicolumn{1}{c|}{0.089} & \multicolumn{1}{c|}{2.633} & 0.991 & \multicolumn{1}{c|}{0.089} & \multicolumn{1}{c|}{3.477} & 0.984 & \multicolumn{1}{c|}{0.091} & \multicolumn{1}{c|}{3.632} & 0.983 \\ \hline
\textbf{Random Forest} & \multicolumn{1}{c|}{\textbf{0.021}} & \multicolumn{1}{c|}{\textbf{0.527}} & \textbf{0.988} & \multicolumn{1}{c|}{\textbf{0.011}} & \multicolumn{1}{c|}{\textbf{0.453}} & \textbf{0.989} & \multicolumn{1}{c|}{\textbf{0.043}} & \multicolumn{1}{c|}{\textbf{0.996}} & \textbf{0.983} & \multicolumn{1}{c|}{\textbf{0.029}} & \multicolumn{1}{c|}{\textbf{0.948}} & \textbf{0.985} & \multicolumn{1}{c|}{\textbf{0.026}} & \multicolumn{1}{c|}{\textbf{0.901}} & \textbf{0.989} \\ \hline
\end{tabular}}
\label{tab:cr-model-comparison-1gb}
\end{table*}

\textbf{Row vs Column Oriented Storage:} Data can be stored in a row oriented fashion, with consecutive row entries stored adjacently, or in a column oriented fashion, with consecutive column entries stored adjacently. It is important to consider how this nuance effects the dynamics for compression ratio predictions. In our experiments, we used CSV files as an example of row storage and Parquet files (common in enterprise data lakes) for columnar storage. Overall, the prediction performance was good in both cases, though the results are slightly better for row storage.

\textbf{Models and Datasets:} We trained several statistical models (XGBoost, Random Forest, SVR) and a Neural Network (MLP) with the features as input to predict compression ratios and decompression speeds. Apart from the naive model of simply averaging, these models performed well and are comparable, while Random Forest performs the best. Tables \ref{tab:cr-model-comparison-1gb},
\ref{tab:cr-model-comparison} and
\ref{tab:ds-model-comparison},
show the results on TPC-H 1GB, TPC-H 100GB and TPC-H 1GB with Zipfian skew. 

\textbf{Sorting Data:}
We briefly investigated how the prediction varies if the data is sorted by different columns. The difference in compression ratios between data sorted by different columns is generally small (of the order of our prediction error). 
We proposed `bucketed weighted entropy' features for capturing the effect of sorting on the entropy of columns. Specifically, the bucketed entropy would be computed for each successive 20\% of rows. Our hypothesis was that for column-store data, there would be a greater change in the compression ratios because entries of a column are stored together, and thus the model should work better for parquet compared to csv files. Empirically, however we observed that prediction performance using the new entropy features was similar to using the older features. We leave further exploration on this as future work.
\begin{table}[htbp]
\caption{Compression Ratio Prediction for Various Models, Compression Schemes, and Data Layouts}
\centering
\resizebox{0.9\linewidth}{!}{%
\begin{tabular}{|c|ccc|ccc|}
\hline \multirow{2}{*}{\textbf{Model}} & \multicolumn{3}{c|}{\textbf{gzip}} & \multicolumn{3}{c|}{\textbf{parquet + gzip}} \\
 & \textbf{MAE} & \textbf{MAPE} & \textbf{R2} & \textbf{MAE} & \textbf{MAPE} & \textbf{R2} \\ \hline\hline
 \multicolumn{7}{|c|}{\textbf{TPC-H 100GB}} \\
 \hline\hline
\textbf{Averaging} & \multicolumn{1}{c|}{0.083} & \multicolumn{1}{c|}{2.378} & - & \multicolumn{1}{c|}{0.324} & \multicolumn{1}{c|}{8.795} & - \\ \hline
\textbf{XGBoost} & \multicolumn{1}{c|}{0.105} & \multicolumn{1}{c|}{2.838} & 0.936 & \multicolumn{1}{c|}{0.151} & \multicolumn{1}{c|}{3.751} & 0.943 \\ \hline
\textbf{Neural Network} & \multicolumn{1}{c|}{0.081} & \multicolumn{1}{c|}{2.232} & 0.968 & \multicolumn{1}{c|}{0.147} & \multicolumn{1}{c|}{3.535} & 0.962 \\ \hline
\textbf{SVR} & \multicolumn{1}{c|}{0.105} & \multicolumn{1}{c|}{3.077} & 0.948 & \multicolumn{1}{c|}{0.19} & \multicolumn{1}{c|}{4.765} & 0.914 \\ \hline
\textbf{Random Forest} & \multicolumn{1}{c|}{\textbf{0.078}} & \multicolumn{1}{c|}{\textbf{2.151}} & \textbf{0.969} & \multicolumn{1}{c|}{\textbf{0.134}} & \multicolumn{1}{c|}{\textbf{3.369}} & \textbf{0.966} \\ \hline\hline
 \multicolumn{7}{|c|}{\textbf{TPC-H Skew}} \\
 \hline\hline
 \textbf{Averaging} & \multicolumn{1}{c|}{0.120} & \multicolumn{1}{c|}{4.915} & - & \multicolumn{1}{c|}{0.601} & \multicolumn{1}{c|}{32.491} & - \\ \hline
\textbf{Neural Network} & \multicolumn{1}{c|}{0.125} & \multicolumn{1}{c|}{3.868} & 0.975 & \multicolumn{1}{c|}{0.336} & \multicolumn{1}{c|}{15.953} & 0.847 \\ \hline
\textbf{SVR} & \multicolumn{1}{c|}{0.101} & \multicolumn{1}{c|}{4.280} & 0.992 & \multicolumn{1}{c|}{0.163} & \multicolumn{1}{c|}{8.526} & 0.969 \\ \hline
\textbf{Random Forest} & \multicolumn{1}{c|}{0.093} & \multicolumn{1}{c|}{3.005} & 0.988 & \multicolumn{1}{c|}{0.251} & \multicolumn{1}{c|}{12.127} & 0.894 \\ \hline
\textbf{XGBoost} & \multicolumn{1}{c|}{\textbf{0.066}} & \multicolumn{1}{c|}{\textbf{2.467}} & \textbf{0.992} & \multicolumn{1}{c|}{\textbf{2.009}} & \multicolumn{1}{c|}{\textbf{6.145}} & \textbf{0.897} \\ \hline
\end{tabular}}
\label{tab:cr-model-comparison}
\end{table}

\textbf{Effect of \textsc{ComPredict} on \textsc{OptAssign}:}
We compare the effect of the predictions 
of compression ratios and decompression times on \textsc{OptAssign}. 
We compute the storage cost, read + compute cost, and latency time of this placement using ground truth compression values as the baseline. The optimization is computed for a range of many different values of $\alpha$ and $\beta$ for comparing  the cost-vs-latency tradeoffs for the predictors (Fig. \ref{fig:compression-optim-tradeoff}). Here we show \textsc{OptAssign} using prediction from SVR on queried samples using weighted entropy features performs very close to ground truth compression for the TPC-H 1GB dataset, not leaving much room for improvement. The magnitude of errors made by our compression predictor is low as seen from the tables shown. In fact, the impact of these errors on the final cost is also minimal since the purple (ground truth compression) and green (our predictor) curves in Fig. \ref{fig:compression-optim-tradeoff} are almost the same. This means both would result in similar latency and storage cost across different tier assignments, unlike other baselines (shown in red and blue).
\begin{figure}[htbp]
    \centering
    \includegraphics[width=\linewidth]{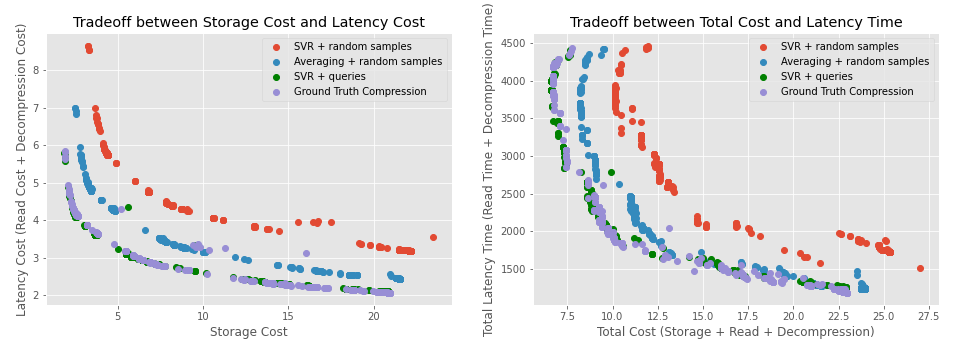}
    \caption{Left: Latency Cost vs Storage Cost, Right: Total Cost vs Latency Time. Different tradeoff curves correspond to different compression predictors used.}
    \label{fig:compression-optim-tradeoff}
\end{figure}
\begin{table}[htbp]
\caption{Decompression (sec/GB) Prediction for Models, Compression Schemes, Data Layouts}
\centering
\resizebox{0.9\linewidth}{!}{%
\begin{tabular}{|c|ccc|ccc|}
\hline
\multirow{2}{*}{\textbf{Model}} & \multicolumn{3}{c|}{\textbf{gzip}} & \multicolumn{3}{c|}{\textbf{parquet + gzip}} \\
 & \textbf{MAE} & \textbf{MAPE} & \textbf{R2} & \textbf{MAE} & \textbf{MAPE} & \textbf{R2} \\ \hline\hline
\multicolumn{7}{|c|}{\textbf{TPC-H 100GB}}\\\hline\hline
\textbf{Averaging} & \multicolumn{1}{c|}{0.679} & \multicolumn{1}{c|}{3.732} & - & \multicolumn{1}{c|}{5.672} & \multicolumn{1}{c|}{43.472} & - \\ \hline
\textbf{XGBoost} & \multicolumn{1}{c|}{0.322} & \multicolumn{1}{c|}{1.773} & 0.972 & \multicolumn{1}{c|}{1.606} & \multicolumn{1}{c|}{10.168} & 0.75 \\ \hline
\textbf{Neural Network} & \multicolumn{1}{c|}{0.147} & \multicolumn{1}{c|}{3.535} & 0.962 & \multicolumn{1}{c|}{1.86} & \multicolumn{1}{c|}{10.875} & 0.522 \\ \hline
\textbf{SVR} & \multicolumn{1}{c|}{0.399} & \multicolumn{1}{c|}{2.153} & 0.961 & \multicolumn{1}{c|}{1.147} & \multicolumn{1}{c|}{10.152} & 0.949 \\ \hline
\textbf{Random Forest} & \multicolumn{1}{c|}{\textbf{0.292}} & \multicolumn{1}{c|}{\textbf{1.601}} & \textbf{0.98} & \multicolumn{1}{c|}{\textbf{1.165}} & \multicolumn{1}{c|}{\textbf{9.698}} & \textbf{0.799} \\ \hline\hline
\multicolumn{7}{|c|}{\textbf{TPC-H Skew}}\\\hline\hline
\textbf{Averaging} & \multicolumn{1}{c|}{7.037} & \multicolumn{1}{c|}{29.979} & - & \multicolumn{1}{c|}{30.134} & \multicolumn{1}{c|}{125.23} & - \\ \hline
\textbf{MLP} & \multicolumn{1}{c|}{1.862} & \multicolumn{1}{c|}{5.860} & 0.917 & \multicolumn{1}{c|}{9.380} & \multicolumn{1}{c|}{21.526} & 0.880 \\ \hline
\textbf{SVR} & \multicolumn{1}{c|}{3.431} & \multicolumn{1}{c|}{15.568} & 0.847 & \multicolumn{1}{c|}{7.020} & \multicolumn{1}{c|}{19.508} & 0.955 \\ \hline
\textbf{XGBoost} & \multicolumn{1}{l|}{2.009} & \multicolumn{1}{l|}{6.145} & \multicolumn{1}{l|}{0.897} & \multicolumn{1}{c|}{6.330} & \multicolumn{1}{c|}{12.284} & 0.948 \\ \hline
\textbf{Random Forest} & \multicolumn{1}{c|}{\textbf{1.141}} & \multicolumn{1}{c|}{\textbf{4.910}} & \textbf{0.922} & \multicolumn{1}{c|}{\textbf{5.194}} & \multicolumn{1}{c|}{\textbf{7.983}} & \textbf{0.915} \\ \hline
\end{tabular}}
\label{tab:ds-model-comparison}
\end{table}
\section{\textsc{DataPart}: Access Aware Data Partitioning}
\label{sec:partition}
Data partitioning in an access pattern aware manner is important for skewed workloads where different parts of a dataset are accessed with widely different frequencies.  

\textsc{DataPart} considers the (minimal) set of records that need to be scanned by a query in an attribute agnostic manner. It then merges these sets of records to generate the data partitions, such that the total scans (read cost) incurred by queries is \textbf{within a limit}, while the overall space required by such partitions is \textbf{minimized} (by reducing overlapping content across partitions). We want to generate balanced size partitions and not fragment the data too much. The decision to make the partitioning attribute agnostic was driven by enterprise data privacy regulations. Fragmenting users’ data across multiple partitions is also undesirable, since that would require that many additional scans of files, hence increasing compute costs (COGS). Informally, our goal is to partition the datasets such that all the files that are generally accessed together belong to the same partition. The files or the contiguous blocks of records need not be adjacent in general (Fig. \ref{fig:partitioning}(a)). 

\subsection{Problem Definition}
Define a query family to comprise of all queries that map to the same files in the data tables.  
Consider a query family $Q$ accessing the following files from $D$: $\{R_1, R_2, \ldots, R_u\}$. This set constitutes an initial (naive) partition of the dataset $P_{Q} = \{R_1, R_2, \ldots, R_u\}$. Let there be $N$ such initial partitions: $\mathcal{P}$, generated from historical access logs. \textsc{DataPart} would generate the final partitions by merging (some of) these initial partitions in order to optimize certain metrics.
We define the span of a partition $P_i \in \mathcal{P}$ as: 
$Sp(P_i):= \sum_{R_k\in P_i}{|R_k|}$, where $|R_k|$ denotes the number of rows or records in the file $R_k$. 
The overlap between two partitions $P_i$ and $P_j$ is the length of files (or, number of records) in common to both partitions, and is computed as 
$Ov(P_i, P_j) = Sp(P_i)+ Sp(P_j) - Sp(P_i\cup P_j)$. 
Span of a merge of partitions $P_i$ and $P_j$, $Sp(P_i\cup P_j) \leq Sp(P_i) + Sp(P_j)$ due to potential overlap between $P_i$ and $P_j$. Formally, a merge $\mathcal{M}$ refers to the union of a set of partitions $\{P_i, P_j, \ldots, P_k\}$ with a span 
$Sp(\mathcal{M}) = Sp(\bigcup_{P_\ell \in \mathcal{M}}P_\ell)$. 

Each partition $P_i$ has an associated access frequency $\rho(P_i)$. The access frequency of a merge is simply the sum of the accesses of the constituent partitions. We require the access frequencies  
of partitions merged should be comparable. Hence we define a feasible merge $\mathcal{M}_k$ as: any pair of partitions $(P_i,\ P_j) \in \mathcal{M}_k$ satisfy at least one of following conditions: (i) $\frac{1}{\rho_{c}} \leq \frac{\rho(P_i)}{\rho(P_j)} \leq \rho_{c}$, or, (ii) $|\rho(P_i) - \rho(P_j)| \leq \rho'_c$, for constants $\rho_c$ and $\rho'_c$. 


Let the `cost' of a merge, measured as the expected read cost, depending on the size and expected number of accesses be defined as: 
$C(\mathcal{M}_k) = Sp(\mathcal{M}_k) \rho(\mathcal{M}_k)$.

The goal is to choose a set of merges $Z$ such that each (initial) partition is part of at least one merge, the total cost of $Z$ is bounded, and the total space required by $Z$ is minimized. 
We formulate this \textsc{Merge Partitions} mathematically as an ILP as follows. Let $\mathcal{M}$ denote 
the set of feasible merges (as defined earlier). In order to ensure that there is 
always a feasible solution, we allow individual partitions also as feasible choices for 
merges. Let $y_k$ be an indicator variable that is $1$ if merge $\mathcal{M}_k$ is chosen in the solution, and $0$ otherwise. 
The first inequality ensures that the (expected) total read cost of all the merges is at most $C_{thresh}$.
Let $x_{\ell, k}$ be another indicator variable that is $1$ if initial partition (or, vertex) $P_\ell$ is covered by merge 
$\mathcal{M}_k$ in the solution, and $0$ otherwise. In other words, this is $1$ when 
$P_\ell \in \mathcal{M}_k$ and $\mathcal{M}_k$ is chosen in the solution, and $0$ otherwise (if a merge $\mathcal{M}_k$ is not part of the solution, 
$x_{\ell, k}$ must be $0$, and this is ensured by the second  inequality). Every partition $P_\ell$ must 
be covered by at least one merge chosen in the solution, and this is ensured by the third inequality. 
\begin{align}
\label{eq:nonlinear}
\nonumber
    \min \sum_{k \in [1, \ldots, |\mathcal{M}|]}{Sp(\mathcal{M}_k) \ y_k} \\\nonumber
   \text{s.t. }\sum_{k \in [1, \ldots, |\mathcal{M}|]}{Sp(\mathcal{M}_k) \rho(\mathcal{M}_k)\  y_k} &\leq C_{thresh} \\\nonumber
     x_{\ell, k} \leq y_k \ \forall P_\ell \in \mathcal{M}_k, \ \forall \mathcal{M}_k \in \mathcal{M}\\\nonumber
     \sum_{\mathcal{M}_j \in \mathcal{M}| P_\ell \in \mathcal{M}_j}{x_{\ell, j}} \geq 1 \ \forall P_\ell \in \mathcal{P} \\\nonumber
     y_k \in {0, 1} \ \forall M_k \in \mathcal{M}, \\
     x_{\ell, k} \in {0, 1} \ \forall P_\ell \in \mathcal{P}
\end{align}

The ILP finds the set 
of merges to minimize the overall space, while covering all segments, keeping cost of merges bounded by $C_{thresh}$. 

\begin{theorem}
\label{thm:np-hard}
\textsc{Merge Partitions} is \textsc{NP-hard}.
\end{theorem}
\begin{proof}
This follows by a reduction from a partitioning problem studied by Huang et al.\cite{orpheus}, where records are shared across versions of datasets, leading to overlap.
Their goal is to divide the set of all versions into different groups, 
and simply store the groups (or, merges), reducing overall storage space and the overall average checkout cost, assuming equal frequency 
of checkout of each version. 
They show the problem of minimizing the checkout cost while keeping the 
storage cost less than a threshold is \textsc{NP-hard}. We construct an instance of \textsc{MergePartitions}, where corresponding to each version, we create a query family (initial partition), and corresponding to each record we create a file, that can be shared across partitions. We want to merge them into groups, reducing overall storage cost, keeping read costs under a threshold. Assuming equal access frequency, it can be seen that the decision version of \textsc{MergePartitions} reduces to the decision version of \textsc{Minimize Checkout Cost} (that is, storage cost $\leq \gamma$ and read cost $\leq C_{thresh}$), shown to be  \textsc{NP-hard} by Huang et al. 
Details omitted due to lack of space.  
\end{proof}

\subsubsection{Algorithm for the General Case: \textsc{G-Part}}
In order to understand the merging problem better, let us consider a graph representation $\mathcal{G} = (\mathcal{V}, \mathcal{E}, \mathcal{W})$ where each initial partition $P_i$ is a node $\in \mathcal{V}$ in graph $\mathcal{G}$. An edge $e = (v, u, w)$ between two vertices $v$ and $u$ in $\mathcal{V}$ with weight $w = \frac{Ov(v,u)}{Sp(v\cup u)}>0$ denotes the 
fractional overlap between the 
partitions $v$ and $u$. ($w=0$ corresponds to no overlap between two partitions, hence,
there is no edge between them). Now, merging can be thought of as merging of nodes to create meta-vertices, collapsing the internal edges, and re-defining edges incident on the meta-vertices from neighbors. (Fig.\ref{fig:partitioning}(c)). 

We give a greedy algorithm \textsc{G-Part} for the general graph case that 
does very well in practice, especially as a key ingredient in the unified pipeline SCOPe and also helps baselines improve significantly. In this algorithm, along with the hard feasibility (based on accesses, as defined earlier) constraints, 
we address a soft constraint $S_{thresh}$ on the span of merges. Specifically, once a merge is $\geq S_{thresh}$, we don't 
merge other partitions to it. The intuition is to prevent the merging 
of too many vertices together, to avoid undue increase in read costs. 

\begin{algorithm}[h]
\caption{\textsc{G-Part}: Partition Merging Algorithm}\label{alg:merging}
\KwData{${P} = \textrm{initial set of partitions}$}
\KwResult{$P = \textrm{new set of partitions after merging}$}
\SetKwFunction{FMerge}{Merge}
$H = [\;]$\tcp*{\small{Max-heap}}
\For{$i \in P$}{
  \For{$j \in P$}{
    \If{$(i, j)$ \textrm{meet merging criteria}}{ 
      $f_{ij} \gets \textrm{fraction of non-overlap}$\;
      $H.push(f_{ij}, i, j)$\;
    }
  }
}
$H.heapify()$\;
$D = \{\}$\tcp*{\small{To store deleted partitions}}
\While{$!H.isempty()$}{
  $f_{ij}, i, j = H.pop()$\;
  \If{$i \in D$ \textbf{or} $j \in D$}{
    \textbf{continue}
  }
  $D.extend\left([i, j]\right)$\;
  $m = \FMerge(i, j)$\tcp*{\small{Merge partition rows}}
  $P.add(m)$\;
  \If{$m.size() < S_{thresh}$}{ 
      \For{$k \in P$}{
        \If{$k \in D$ \textbf{or} $k == m$}{
        \textbf{continue}\;
        }
        \If{$(m, k)$ \textrm{meet merging criteria}}{
            $H.push\left(f_{mk}, m, k\right)$\;
          }
      }
  }
}
\For{$i \in D$}{
  $P.remove(i)$
}
\end{algorithm}
We next describe \textsc{G-Part} informally. (The pseudocode is given in Algorithm \ref{alg:merging}.) 
We first filter the edges to determine the set of feasible edges. 
We store the edges in a 
max-heap, where the heapification is on the weights (denoting the fractional overlaps) 
of the edges. We pick the top most heap element edge (this has the highest fractional 
overlap between the pair of segments) and merge the corresponding 
pair of partitions. Let $u$ and $v$ be the corresponding nodes. 
We create a new (merge) node $u'$, while removing $u$ and $v$ from $\mathcal{V}$. 
$\mathcal{V}$ is updated as $u' \cup \{\mathcal{V}\setminus \{u, v\}\}$. 
Similarly, the edge $e_{v,u}$ is deleted 
from $\mathcal{E}$. 
If the span of the merged node $Sp(u') \geq S_{thresh}$, for some 
constant $S_{thresh}$, then we don't consider 
the merged node any further. Specifically, we remove every edge $e' = (w,x)$, 
where $x \in \{u, v\})$, for any $w \in \mathcal{V}$ from $\mathcal{E}$, and delete these edges from 
the heap. However, if $Sp(u') < S_{thresh}$, then it goes back as a candidate 
for further merging. In this case, for 
every edge $e' = (w,x)$ for any $x \in \{u,v\}$ and $w \in \mathcal{V}$, we replace 
it with $e'' = (w, u')$ in $\mathcal{E}$ (and delete $e'$ from 
the heap, if it was present in the heap). If $e''$ satisfies the feasibility constraints, 
we add it to the heap with a weight corresponding to the fractional overlap of $u'$ with $w$. Now, we repeat the process with the next top heap element, 
till the heap is empty. Note that at the end we might be left with singleton 
partitions that do not meet the feasibility constraints for merging.  
These are (individually) added to the set of final merges or partitions.

\begin{figure*}
\caption{(a) Data partitioning examples. (b) Bipartite matching for equal sized partitions with no compression. (c) Merging of nodes by \textsc{G-Part} in a graph setting. }
\centering
\begin{minipage}[b]{.34\textwidth}
    \includegraphics[width=\linewidth, height=0.6\linewidth]{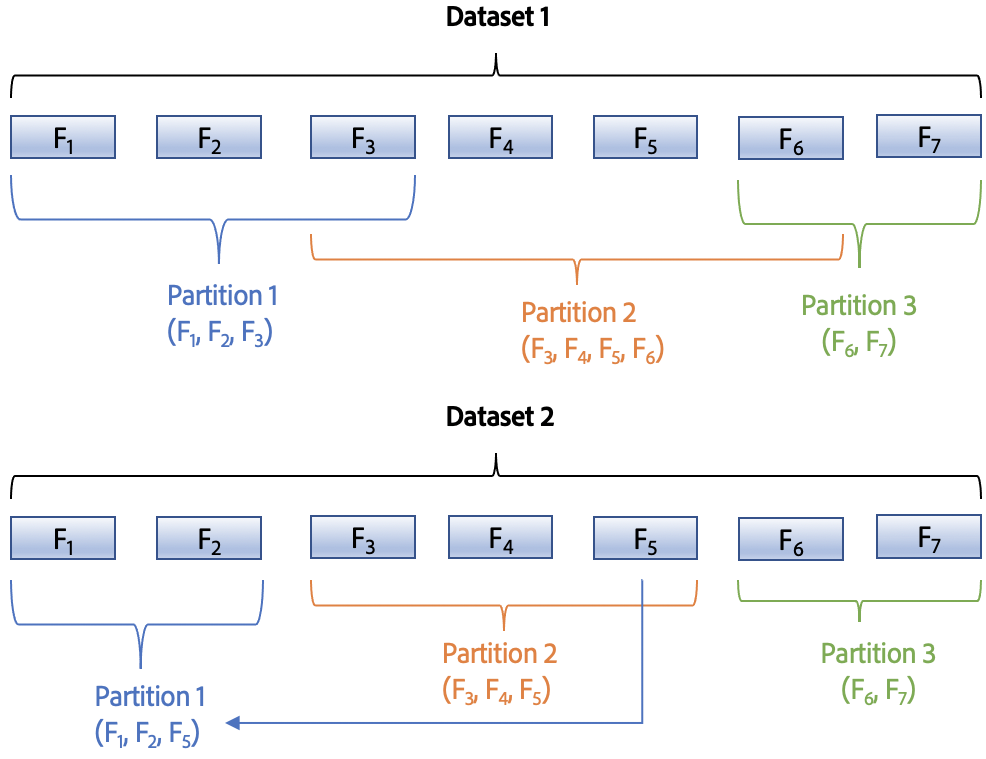}
    \subcaption{\small{}}
\end{minipage}
\hfill
\begin{minipage}[b]{.31\textwidth}
    \centering
\includegraphics[width=0.7\linewidth, height=0.645\linewidth]{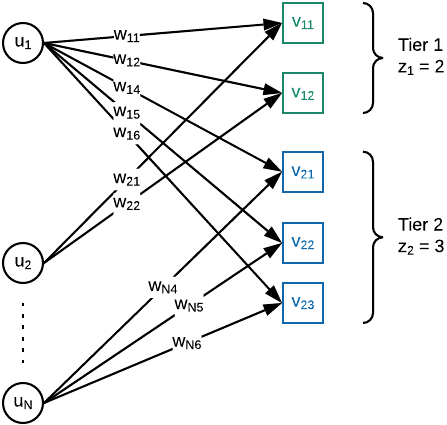}
  \subcaption{\small{}}
\end{minipage}%
\hfill
\begin{minipage}[b]{.32\textwidth}
  \centering\includegraphics[width=\linewidth, height=0.5\linewidth]{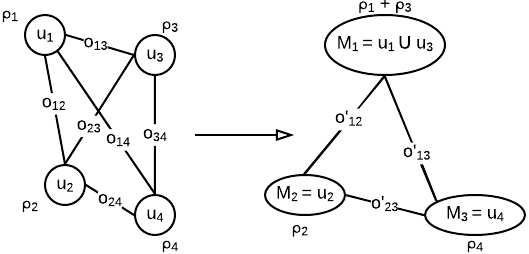}
  \subcaption{\small{}}
\end{minipage}
\label{fig:partitioning}
\end{figure*}

\textbf{Space and Cost trade-off achieved by \textsc{G-Part}:}
We evaluate \textsc{G-Part} on TPC-H 1GB and TPC-H 100GB to compare the duplication of data with the cost of merging in Fig. \ref{fig:part-tradeoff}\footnote{Duplication is computed as $1 - \frac{|\{P_i\}|}{|P_i|}$, where $\{P_i\}$ denotes the set of distinct elements (records) in $P_i$. The cost of merging is computed as the increase in expected read cost due to merging.}.
\textsc{G-Part} provides a good trade-off between the unmerged case and merging all partitions. 

\textbf{Complexity of \textsc{G-Part}}: Consider there are $m$ query families. Let the number of files across all datasets in a client org be $n$. Initial processing to generate initial partitions would require a space of at most $O(m n)$ (In general it would be much less, say, $O(m k)$ where $k$ is the average number of files accessed by each query family). There would be $m$ initial partitions, and $O(m^2)$ edges. For estimating the cost of each edge, one would need to find the intersection of the sets (of files in each partitions) and the total length of each partition (sum of the size of the files, maintained as file meta data). The heap construction followed by \textsc{G-Part} merging and heapifying would take $O(m^2 \log m)$. This can adapt dynamically to changing query workloads. For each new query family observed in the workload, we create an initial partition, by labeling the query family with the accessed files. This results in a new node in the graph. If the current number of merged partitions is $m' \ll m$, at most $O(m')$ new edges get added. This can result in $O(m' \log m')$ operations for heapifying followed by merging operations.

\subsection{Special Case: Ordered Partitions}
Consider an inherent ordering between the files, such as that arising for time series data. Let us  
assume that the data is time stamped. Each query, and 
hence partition $P_i$ has a fixed start time $s(P_i)$ and a fixed end time $e(P_i)$. Let us order the partitions by their end times. We only consider distinct queries. 
Let this ordered set be $\mathcal{P}$, 
where $|\mathcal{P}| = N$.
Partition $P_i \in \mathcal{P}$, has the $i^{th}$ latest end time and $P_1$ has the first end time.

Since the main motivation in merging is exploiting the overlap between 
partitions or segments, we only consider combinations of adjacent segments in the order in which they occur 
in the ordered list. 
More specifically, consider partitions $\{P_i, P_{i+1}, P_{i+2}\} \in \mathcal{P}$. 
The possible set of (merged) segments  
corresponding to these are: 
(i) $\{P_i\}, \{P_{i+1}\}, \{P_{i+2}\}$, or, 
(ii) $\{P_i, P_{i+1}\}, \{P_{i+2}\}$, or, 
(iii) $\{P_i\}, \{P_{i+1}, P_{i+2}\}$, or, 
(iv) $\{P_i, P_{i+1}, P_{i+2}\}$, 
\textbf{but not} $[\{P_{i}, p_{i+2}\}, \{p_{i+1}\}]$. 

The number of possible merges, that is, $|\mathcal{M}|$ is $O(N^2)$. 
Without loss of generality, we assume that for every $P_i$, with end time $e(P_i)$, the start time of 
$P_{i+1}$, $s(P_{i+1}) < e(P_i)$. (For any pair of $i$ and $i+1$ where this does not hold, we can consider the 
set of partitions $\{P_1, \ldots, P_i\}$ and $\{P_{i+1}, \ldots, P_{N}\}$ to be disjoint and solve 
the merging separately for each set.)

We define a dynamic program here. Consider the sub-problem of covering partitions 
$[P_1, \ldots, P_i]$. 
Define the set of feasible merges containing partition $P_i$ as $\mathcal{F}_i$. 
As defined earlier, 
$\mathcal{F}_i = \{[P_1, P_2, \ldots, P_i], [P_2, \ldots, P_i], \ldots, [P_i]\}$. 
Any feasible solution 
on partitions $[P_1, \ldots, P_i]$ must include a merge in $\mathcal{F}_i$. For ease of analysis, WLOG, we add a dummy partition $P_0$ of $Sp(P_0) = 0$ and $\rho(P_0) = 0$. 

Let us denote the merge $[P_{i-k}, \ldots, P_i]$ as $\mathcal{M}^k_i$ for $k \in \{0,1,\ldots, i-1\}$. 
We define the parent of $\mathcal{M}^k_i$ as $P(\mathcal{M}^k_i) := P_{i-k-1}$ for $k \leq i-1$ ($P_0$ for $k=i-1$). The notion of parent simply implies that if a feasible solution 
chooses $\mathcal{M}^k_i$, then
(i) it must include additional merges to cover 
$[P_0, \ldots, P(\mathcal{M}^k_i)]$, and 
such a solution must fit 
within the remaining cost budget after the choice of $\mathcal{M}^k_i$. The recurrence relations are: 

For $i=0$, $ALG[P_0, C] = 0\ \forall C\geq 0$.

For $i>0$ and $0\leq C \leq C_{thresh}$, \\
$ALG[P_i, C] = \min_{k \in [0, \ldots, i-1]| C(\mathcal{M}^k_i) \leq C}\ ALG[(P(\mathcal{M}^k_i)), C - C(\mathcal{M}^k_i)] + Sp(\mathcal{M}^k_i) \  \forall\  i\in [N],\ \forall \  0<c\leq C$. 

\begin{theorem}
\label{thm:alg-opt}
$ALG(P_N, C_{thresh})$ \textbf{minimizes the overall space}  given a budget $C_{thresh}$ on the total (expected) read cost .
\end{theorem}
Proof omitted due to lack of space, however it follows by induction, based on induction hypothesis, after proving the  optimality of base cases. The time complexity of $ALG$ is $O(N^2 C_{thresh})$, which makes it pseudo-polynomial solution because of the dependence on $C_{thresh}$. 
To get a polynomial approximation scheme, we bucket the range of $C_{thresh}$.

\begin{theorem}
\label{thm:polyapprox}
Let the optimal solution for $N$ partitions with a cost threshold $C_{OPT}$ require space
$OPT[N,C_{OPT}] = S_{OPT}$. Then there exists a polynomial algorithm that finds a solution of space $\leq S_{OPT}$, 
within a cost at most $(1+N\epsilon) C_{OPT}$ in $O(N^2(N + \frac{1}{\epsilon}))$ time for any fixed $\epsilon >0 $.
For $\epsilon = \frac{1}{N}$, we get a $(1, 2)$ bi-criteria 
approximation of $(S_{OPT}, C_{OPT})$ in $O(N^3)$. 
\end{theorem}

Proof is omitted due to lack of space. The main idea is to discretize the range of cost values by rounding up by $\epsilon$, extending the cost threshold by $N\epsilon$, and solving ALG on this setting. It can be argued that extending the cost threshold by $N\epsilon$ would ensure a feasible solution exists. By optimality of ALG, the space required by ALG would be minimum, hence $\leq S_{OPT}$, and by feasibility of ALG, the total cost would be bounded by $\leq (1+N\epsilon) C_{OPT}$. For $\epsilon \leq 1/N$, the cost is $\leq 2 C_{OPT}$, giving the $(1,2)$ bi-criteria solution.

\begin{figure}[htbp]
\caption{Space cost tradeoffs in  partitioning. Each dot in the scatterplot represents a table. We consider 3 cases - (i) no merging, (ii) \textsc{G-Part} heuristic, and (iii) merging all partitions. Left: TPC-H 1GB, Right: TPC-H 100GB.}
\label{fig:part-tradeoff}
\centering
\begin{minipage}[b]{.48\linewidth}
  \includegraphics[width=\linewidth]{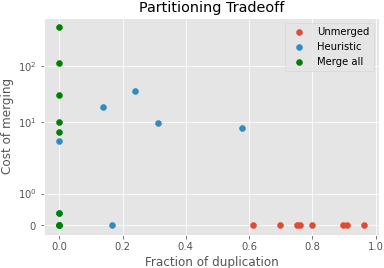}
\end{minipage}
\hfill
\begin{minipage}[b]{.49\linewidth}
    \includegraphics[width=\linewidth]{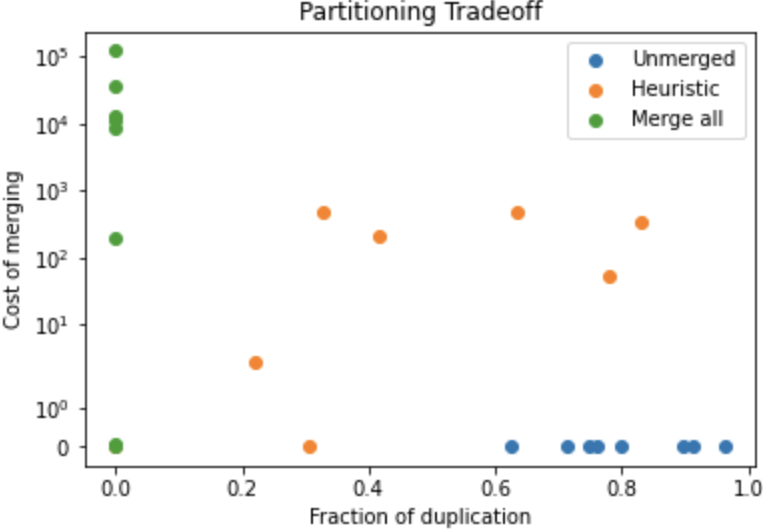}
\end{minipage}
\end{figure}
\section{Unified pipeline \textsc{SCOPe}}
\label{sec:experiments}
Here we present the unified pipeline SCOPe that combines all the modules \textsc{OptAssign}, \textsc{ComPredict} and \textsc{G-Part} to optimize the overall costs, while maintaining performance guarantees. We use TPC-H 100GB, TPC-H 1TB datasets and Enterprise datasets I for these experiments. 
The pipeline is as follows. 
First we generate initial partitions using query logs. These are merged using \textsc{G-Part} to generate final partitions. After this, \textsc{ComPredict} predicts the compression ratio and decompression speeds for each partition. Finally, \textsc{OptAssign} finds optimal tier and compression scheme assignment for the partitions, minimizing the overall costs, including storage and read costs, subject to capacity constraints and latency SLAs. 

\begin{table*}[htbp]
\caption{Results for Enterprise Data II.}
\resizebox{\textwidth}{!}{%
\begin{tabular}{|c|c|c|c|c|c|c|c|c|c|c|c|}
\hline
\textbf{\multirow{2}{*}{Variants we can support}} & \textbf{Other methods} & \textbf{\multirow{2}{*}{P}} & \textbf{\multirow{2}{*}{T}} & \textbf{\multirow{2}{*}{C}} & \textbf{Storage} & \textbf{Decomp.} & \textbf{Read} & \textbf{Total} & \textbf{Read Latency} & \textbf{Expected Decomp.} & \textbf{Tiering} \\
 & \textbf{we can adapt} &  &  &  & \textbf{Cost} & \textbf{Cost} & \textbf{Cost} & \textbf{Cost} & \textbf{(TTFB, s)} & \textbf{Latency (ms)} & \textbf{Scheme} \\ \hline\hline
Default (store on premium) & - & - & - & - & 150.1 & 0.0	& 18.74 &	168.9 &	0.024 & 0.0 & {[}3, 0, 0{]} \\ \hline
Compress \& store on premium & Ares & - & - & Y & 138.8 & 0.1 & 18.5 & 157.4 & 0.024 & 0.016 & {[}3, 0, 0{]} \\ \hline
Multi-Tiering & Hermes & - & Y & - & 20 & 0.0 & 62 & 82 & 0.281 & 0.0 & {[}0, 2, 1{]} \\ \hline
Latency time focused & HCompress & - & Y & Y & 49.6 & 0.0 & 49.4 & 98.9 & 0.165 & 0.0 & {[}2, 1, 0{]} \\ \hline\hline
Partition \& store on premium & - & Y & - & - & 102.7 & 0.0 & 1.2 & 103.9 & 0.024 & 0.0 & {[}23, 0, 0{]} \\ \hline
Partitioning + Tiering & Hermes + \textsc{G-Part} & Y & Y & - & 36.3 & 0.0 & 26.7 & 62.9 & 0.281 & 0.0 & {[}0, 4, 19{]} \\ \hline
Partitioning + Compression & Ares + \textsc{G-Part} & Y & - & Y & 130.1 & 0.8 & 2.3 & 133.1 & 0.024 & 0.170 & {[}23, 0, 0{]} \\ \hline
\hline
\textbf{SCOPe (Latency time focused)} & \textbf{HCompress + \textsc{G-Part}} & \textbf{Y} & \textbf{Y} & \textbf{Y} & \textbf{94.9} & \textbf{0.0} & \textbf{26.4} & \textbf{121.2} & \textbf{0.164} & \textbf{0.0001} & \textbf{{[}16, 3, 4{]}} \\ \hline
\textbf{SCOPe (No capacity constraint)} & \textbf{-} & \textbf{Y} & \textbf{Y} & \textbf{Y} & \textbf{22.7} & \textbf{0.6} & \textbf{7.0} & \textbf{30.3} & \textbf{0.216} & \textbf{0.131} & \textbf{{[}2, 11, 10{]}} \\ \hline
\textbf{SCOPe (Read+Decomp. cost focused)} & \textbf{-} & \textbf{Y} & \textbf{Y} & \textbf{Y} & \textbf{75.5} & \textbf{0.5} & \textbf{5.2} & \textbf{81.2} & \textbf{0.084} & \textbf{0.110} & \textbf{{[}6, 15, 2{]}} \\ \hline
\textbf{SCOPe (Total cost focused)} & \textbf{-} & \textbf{Y} & \textbf{Y} & \textbf{Y} & \textbf{22.7} & \textbf{0.6} & \textbf{7.0} & \textbf{30.3} & \textbf{0.216} & \textbf{0.131} & \textbf{{[}2, 11, 10{]}} \\ \hline
\end{tabular}}
\label{tab:optim-private}
\end{table*}
\begin{table*}[htp]
\caption{Results for the TPC-H dataset (100GB).}
\resizebox{\textwidth}{!}{%
\begin{tabular}{|c|c|c|c|c|c|c|c|c|c|c|c|}
\hline
\textbf{\multirow{2}{*}{Variants we can support}} & \textbf{Other methods} & \textbf{\multirow{2}{*}{P}} & \textbf{\multirow{2}{*}{T}} & \textbf{\multirow{2}{*}{C}} & \textbf{Storage} & \textbf{Decomp.} & \textbf{Read} & \textbf{Total} & \textbf{Read Latency} & \textbf{Expected Decomp.} & \textbf{Tiering} \\
 & \textbf{we can adapt} &  &  &  & \textbf{Cost} & \textbf{Cost} & \textbf{Cost} & \textbf{Cost} & \textbf{(TTFB, s)} & \textbf{Latency (ms)} & \textbf{Scheme} \\ \hline\hline
Default (store on premium) & - & - & - & - & 8741.9	& 0.0 & 	3828.5	& 12570.4	& 0.18 & 0.0 & {[}8, 0, 0{]} \\ \hline
Compress \& store on premium & Ares & - & - & Y & 7138.2 &	121.1 & 3387.5 & 10646.8 & 0.18 & 3.61 & {[}8, 0, 0{]} \\ \hline
Multi-Tiering & Hermes & - & Y & - & 8741.8 & 0.0 & 3828.5 & 12570.4 & 0.18 & 0.0 & {[}5, 3, 0{]} \\ \hline
Latency time focused & HCompress & - & Y & Y & 3288.4 & 0.0 & 22805.0 & 26093.4 & 0.68 & 0.0 & {[}7, 0, 1{]} \\ \hline\hline
Partition \& store on premium & - & Y & - & - & 8702.6 & 0.0 & 117.3 & 8819.9 & 0.18 & 0.0 & {[}137, 0, 0{]} \\ \hline
Partitioning + Tiering & Hermes + \textsc{G-Part} & Y & Y & - & 1397.0 & 0.0 & 415.3 & 1812.4 & 2.06 & 0.0 & {[}0, 94, 43{]} \\ \hline
Partitioning + Compression & Ares + \textsc{G-Part} & Y & - & Y & 5480.4 & 32.1 & 60.9 & 5573.4 & 0.18 & 0.96 & {[}137, 0, 0{]} \\ \hline
\hline
\textbf{SCOPe (Latency time focused)} & \textbf{HCompress + \textsc{G-Part}} & \textbf{Y} & \textbf{Y} & \textbf{Y} & \textbf{5178.1} & \textbf{0.0} & \textbf{544.5} & \textbf{5722.6} & \textbf{0.48} & \textbf{0.0} & \textbf{{[}108, 0, 29{]}} \\ \hline
\textbf{SCOPe (No capacity constraint)} & \textbf{-} & \textbf{Y} & \textbf{Y} & \textbf{Y} & \textbf{691.4} & \textbf{29.9} & \textbf{219.3} & \textbf{940.6} & \textbf{2.06} & \textbf{0.89} & \textbf{{[}0, 94,43{]}} \\ \hline
\textbf{SCOPe (Read+Decomp cost focused)} & \textbf{-} & \textbf{Y} & \textbf{Y} & \textbf{Y} & \textbf{4733.9} & \textbf{17.4} & \textbf{80.9} & \textbf{4832.1} & \textbf{0.35} & \textbf{0.52} & \textbf{{[}103, 34, 0{]}} \\ \hline
\textbf{SCOPe (Total cost focused)} & \textbf{-} & \textbf{Y} & \textbf{Y} & \textbf{Y} & \textbf{679.2} &\textbf{31.1} & \textbf{242.4} & \textbf{952.7} & \textbf{2.06} & \textbf{0.93} & \textbf{{[}0, 82, 55{]}}\\ \hline
\end{tabular}}
\label{tab:optim-100}
\end{table*}

\begin{table*}[htp]
\caption{Results for TPC-H dataset (1TB), (K refers to a multiplicative factor of $10^3$.)}
\resizebox{\textwidth}{!}{%
\begin{tabular}{|c|c|c|c|c|c|c|c|c|c|c|c|}
\hline
\textbf{\multirow{2}{*}{Variants we can support}} & \textbf{Other methods} & \textbf{\multirow{2}{*}{P}} & \textbf{\multirow{2}{*}{T}} & \textbf{\multirow{2}{*}{C}} & \textbf{Storage} & \textbf{Decomp.} & \textbf{Read} & \textbf{Total} & \textbf{Read Latency} & \textbf{Expected Decomp.} & \textbf{Tiering} \\
 & \textbf{we can adapt} &  &  &  & \textbf{Cost} & \textbf{Cost} & \textbf{Cost} & \textbf{Cost} & \textbf{(TTFB, s)} & \textbf{Latency (ms)} & \textbf{Scheme} \\ \hline\hline
Default (store on premium) & - & - & - & - & 89.23K	& 0.0 & 39.13K & 128.36K & 0.18 & 0.0 & {[}8, 0, 0{]} \\ \hline
Compress \& store on premium & Ares & - & - & Y & 73.79K	& 3.36K	& 34.85K	& 112.01K & 0.18 & 100.31 & {[}8, 0, 0{]} \\ \hline
Multi-Tiering & Hermes & - & Y & - & 89.11K & 0.0 & 38.94K & 128.05K & 0.18 & 0.0 & {[}5, 3, 0{]} \\ \hline
Latency time focused & HCompress & - & Y & Y & 41.58K & 0.0 & 242.47K & 284.05K & 1.07 & 0.0 & {[}6, 2, 0{]} \\ \hline\hline
Partition\& store on premium & - & Y & - & - & 81.37K & 0.0 & 3.16K & 84.53K & 0.18 & 0.0 & {[}212, 0, 0{]} \\ \hline
Partitioning + Tiering & Hermes + \textsc{G-Part} & Y & Y & - & 26.77K & 0.0 & 7.51K & 34.28K & 2.91 & 0.0 & {[}0, 148, 64{]} \\ \hline
Partitioning + Compression & Ares + \textsc{G-Part} & Y & - & Y & 47.05K & 2.20K & 1.13K & 50.38K & 0.18 & 65.68 & {[}212, 0, 0{]} \\ \hline
\hline
\textbf{SCOPe (Latency time focused)} & \textbf{HCompress + \textsc{G-Part}} & \textbf{Y} & \textbf{Y} & \textbf{Y} & \textbf{64.68K} & \textbf{0.0} & \textbf{4.76K} & \textbf{69.44K} & \textbf{1.44} & \textbf{0.0} & \textbf{{[}101, 77, 34{]}} \\ \hline
\textbf{SCOPe (No capacity constraint)} & \textbf{-} & \textbf{Y} & \textbf{Y} & \textbf{Y} & \textbf{17.93K} & \textbf{1.03K} & \textbf{6.46K} & \textbf{25.42K} & \textbf{2.91} & \textbf{30.89} & \textbf{{[}0, 176, 36{]}} \\ \hline
\textbf{SCOPe (Read+Decomp cost focused)} & \textbf{-} & \textbf{Y} & \textbf{Y} & \textbf{Y} & \textbf{61.30K} & \textbf{0.78K} & \textbf{1.66K} & \textbf{63.74K} & \textbf{1.15} & \textbf{23.32} & \textbf{{[}89, 123, 0{]}} \\ \hline
\textbf{SCOPe (Total cost focused)} & \textbf{-} & \textbf{Y} & \textbf{Y} & \textbf{Y} & \textbf{15.14K} &\textbf{0.12K} & \textbf{4.53K} & \textbf{19.79K} & \textbf{3.20} & \textbf{36.63} & \textbf{{[}0, 155, 57{]}}\\ \hline
\end{tabular}}
\label{tab:optim-1T}
\end{table*}

\textbf{Comparison with Baselines: }
To the best of our knowledge, SCOPe as a pipeline is unique, and there are no direct baselines we could compare with. However, by tuning the parameters of \textsc{OptAssign}, we can choose to optimize either only tiering and no compression ($K=0$), or, only compression and no tiering ($L=0$), or, minimize the latency due to read costs and decompression costs ($\alpha = 0$). These variants would map to an adaption of existing storage optimization approaches like HCompress\cite{devarajan2020hcompress} (focused on reducing latency), Hermes\cite{hermes} (focused on multi-tiering), and Ares\cite{ares} (focused on compression), which were originally designed for different settings with I/O workloads. Hence, these can be thought of as our baselines, adapted to our setting. We can see that overall, SCOPe performs very well and minimizes the costs while maintaining good trade-offs on the different costs and latency. Moreover, we show that applying \textsc{G-Part} to generate data partitions before applying the baseline methods, significantly improves the performance of baselines\footnote{All results are generated using ground truth compression data ensuring a fair comparison.}. 
The optimization takes about 47.4 ms on average (min 35 ms, max 470 ms) for optimization given one set of hyperparameters. Tuning the hyperparameters for optimization takes $\approx 18.9$s. 

All costs of Tables \ref{tab:optim-private}, \ref{tab:optim-100} and \ref{tab:optim-1T} are calculated over a 5.5 month duration using Azure cost parameters. We considered the cost savings opportunity by only considering Premium, Hot and Cool Layers (and not Archive, since that has an early deletion period of 6 months. We have examined Archive benefits in our Enterprise Data I experiments described earlier). 
Next we explain the structure of the tables and the results. 
The rows refer to policies. The column ('Other methods ...') refers to the closest baseline in the literature. The last column `Tiering Scheme' refers to the to the number of partitions (or, datasets) assigned to tiers [Premium, Hot, Cool] respectively. The other columns are self-explanatory. 

The first 4 rows of the table focus on standard approaches that are generally followed and these typically have higher total costs. Row 1 and 2 store everything on premium, generally incurs low read costs and high storage costs. Row 3 `Multi-Tiering' refers to optimal multi-tiering and incurs much lower storage costs, but read costs and read latencies are higher. Since there is no compression, there is no decompression cost or latency. Row 4 `Latency time focused' aims at minimizing the storage costs with a focus on keeping total latency low. Rows 5-7 employ \textsc{G-Part} on top of Rows 1, 3 and 2. We can see that the total costs are significantly lower now, while the latencies remain comparable. Also note that the total number of partitions has increased from the original number of datasets (as observed from the last column). 
\begin{table}[htbp]
\caption{Parameters for ILP Optimization on TPC-H.} 
\resizebox{\linewidth}{!}{%
\begin{tabular}{|l|l|l|l|l|}
\hline
\textbf{Parameters} & \textbf{Premium}   & \textbf{Hot} & \textbf{Cool} & \textbf{Archive} \\ \hline
Storage cost $C^s_\ell$ (cents/GB)     & 15        & 2.08     & 1.52    & 0.099   \\ \hline
Read cost $C^r_\ell$ (cents/GB)       & 0.004659 & 0.01331 & 0.0333 & 16.64   \\ \hline
Layer capacity $S_\ell$ (GB) & 0.163 & 0.326 & 0.4891 & inf \\ \hline
Read latency or TTFB (Time to first byte) $B_\ell$ (sec) & 0.0053    & 0.0614   & 0.0614  & 3600    \\ \hline
compute cost $C^c$ (cents/sec)  & \multicolumn{4}{c|}{0.001}\\ \hline
\end{tabular}}
\label{tab:params}
\end{table}
The last 4 rows illustrate the different variants of the entire pipeline SCOPe (with partitioning, multi-tiering and compression). These highlight how SCOPe can be tuned based on user requirements. We can see SCOPe optimizes the respective objectives while maintaining a very good trade-off on other metrics. SCOPe consistently performs well across all datasets at different scales, namely, enterprise, TPC-H 1GB (not shown here), 100GB and 1TB. `SCOPe (Total cost focused)' is consistently within $8$ - $18\%$ of the `Default' (platform baseline), and incurs the lowest cost among all the other baselines and variants.
Read latency (Time to First Byte) does not change significantly since it is independent of size (Table \ref{tab:params}). The expected decompression latency (average across accesses) and other costs grow with the size of partitions.
\section{Conclusion and Future Work}
\label{sec:conclusion}
We present SCOPe: a tunable framework that optimizes storage and access costs
on the cloud while maintaining latency guarantees. It is substantially better than baselines and works extremely well across different types and scales of data, giving cost benefits of the order of ~50\% or greater. 
Going forward, we want to extend SCOPe to optimize compute costs, including recommending optimal configurations. 

\textbf{Acknowledgements: } We would like to thank our colleages Shone Sadler, Dilip Biswal and Daniel Sirbu from Adobe Experience Platform for helping us get access to the data, costs and infrastructure for performing the experiments. We would especially thank Shone for the continual encouragement, guidance and help in this project. 
\bibliographystyle{plain}
\bibliography{ref}

\begin{thebibliography}{10}

\bibitem{abadi2006integrating}
Daniel Abadi, Samuel Madden, and Miguel Ferreira.
\newblock Integrating compression and execution in column-oriented database
  systems.
\newblock In {\em Proceedings of the 2006 ACM SIGMOD international conference
  on Management of data}, pages 671--682, 2006.

\bibitem{succinct}
Rachit Agarwal, Anurag Khandelwal, and Ion Stoica.
\newblock Succinct: Enabling queries on compressed data.
\newblock In {\em 12th USENIX Symposium on Networked Systems Design and
  Implementation (NSDI 15)}, pages 337--350, 2015.

\bibitem{blinkdb}
Sameer Agarwal, Barzan Mozafari, Aurojit Panda, Henry Milner, Samuel Madden,
  and Ion Stoica.
\newblock Blinkdb: queries with bounded errors and bounded response times on
  very large data.
\newblock In {\em Proceedings of the 8th ACM European Conference on Computer
  Systems}, pages 29--42, 2013.

\bibitem{anwar2015taming}
Ali Anwar, Yue Cheng, Aayush Gupta, and Ali~R Butt.
\newblock Taming the cloud object storage with mos.
\newblock In {\em Proceedings of the 10th Parallel Data Storage Workshop},
  pages 7--12, 2015.

\bibitem{anwar2016mos}
Ali Anwar, Yue Cheng, Aayush Gupta, and Ali~R Butt.
\newblock Mos: Workload-aware elasticity for cloud object stores.
\newblock In {\em Proceedings of the 25th ACM International Symposium on
  High-Performance Parallel and Distributed Computing}, pages 177--188, 2016.

\bibitem{aws-tiering}
AWS.
\newblock {Amazon S3 Intelligent-Tiering storage class}.
\newblock \url{https://aws.amazon.com/s3/storage-classes/intelligent-tiering/},
  2022.
\newblock [Online; accessed 6-Jan-2023].

\bibitem{azure-lifecycle}
Azure.
\newblock {Azure Blob Storage lifecycle management generally available}.
\newblock
  \url{https://azure.microsoft.com/en-in/blog/azure-blob-storage-lifecycle-management-now-generally-available/},
  2019.
\newblock [Online; accessed 6-Jan-2023].

\bibitem{azure}
Azure.
\newblock {Azure Data Lake Storage pricing}.
\newblock
  \url{https://azure.microsoft.com/en-in/pricing/details/storage/data-lake/},
  2022.
\newblock [Online; accessed 8-October-2022].

\bibitem{brendle2022sahara}
Michael Brendle, Nick Weber, Mahammad Valiyev, Norman May, Robert Schulze,
  Alexander B{\"o}hm, Guido Moerkotte, and Michael Grossniklaus.
\newblock Sahara: Memory footprint reduction of cloud databases with automated
  table partitioning.
\newblock In {\em EDBT}, pages 1--13, 2022.

\bibitem{roaring1}
Samy Chambi, Daniel Lemire, Owen Kaser, and Robert Godin.
\newblock Better bitmap performance with roaring bitmaps.
\newblock {\em Software: practice and experience}, 46(5):709--719, 2016.

\bibitem{cheng2015cast}
Yue Cheng, M~Safdar Iqbal, Aayush Gupta, and Ali~R Butt.
\newblock Cast: Tiering storage for data analytics in the cloud.
\newblock In {\em Proceedings of the 24th International Symposium on
  High-Performance Parallel and Distributed Computing}, pages 45--56, 2015.

\bibitem{concise}
Alessandro Colantonio and Roberto Di~Pietro.
\newblock Concise: Compressed ‘n’composable integer set.
\newblock {\em Information Processing Letters}, 110(16):644--650, 2010.

\bibitem{PLWAH}
Fran{\c{c}}ois Deli{\`e}ge and Torben~Bach Pedersen.
\newblock Position list word aligned hybrid: optimizing space and performance
  for compressed bitmaps.
\newblock In {\em Proceedings of the 13th international conference on Extending
  Database Technology}, pages 228--239, 2010.

\bibitem{devarajan2020hcompress}
Hariharan Devarajan, Anthony Kougkas, Luke Logan, and Xian-He Sun.
\newblock Hcompress: Hierarchical data compression for multi-tiered storage
  environments.
\newblock In {\em 2020 IEEE IPDPS}, pages 557--566. IEEE, 2020.

\bibitem{ares}
Hariharan Devarajan, Anthony Kougkas, and Xian-He Sun.
\newblock An intelligent, adaptive, and flexible data compression framework.
\newblock In {\em 2019 19th IEEE/ACM International Symposium on Cluster, Cloud
  and Grid Computing (CCGRID)}, pages 82--91. IEEE, 2019.

\bibitem{devarajan2020hfetch}
Hariharan Devarajan, Anthony Kougkas, and Xian-He Sun.
\newblock Hfetch: Hierarchical data prefetching for scientific workflows in
  multi-tiered storage environments.
\newblock In {\em 2020 IEEE IPDPS}, pages 62--72. IEEE, 2020.

\bibitem{mto_kraska}
Jialin Ding, Umar~Farooq Minhas, Badrish Chandramouli, Chi Wang, Yinan Li, Ying
  Li, Donald Kossmann, Johannes Gehrke, and Tim Kraska.
\newblock Instance-optimized data layouts for cloud analytics workloads.
\newblock In {\em Proceedings of the 2021 International Conference on
  Management of Data}, SIGMOD '21, page 418–431, New York, NY, USA, 2021.
  Association for Computing Machinery.

\bibitem{ERRADI2020110457}
Abdelkarim Erradi and Yaser Mansouri.
\newblock Online cost optimization algorithms for tiered cloud storage
  services.
\newblock {\em Journal of Systems and Software}, 160:110457, 2020.

\bibitem{VALWAH}
Gheorghi Guzun, Guadalupe Canahuate, David Chiu, and Jason Sawin.
\newblock A tunable compression framework for bitmap indices.
\newblock In {\em 2014 IEEE 30th international conference on data engineering},
  pages 484--495. IEEE, 2014.

\bibitem{kinoshita2021cost}
Reika Kinoshita, Satoshi Imamura, Lukas Vogel, Satoshi Kazama, and Eiji
  Yoshida.
\newblock Cost-performance evaluation of heterogeneous tierless storage
  management in a public cloud.
\newblock In {\em 2021 Ninth International Symposium on Computing and
  Networking (CANDAR)}, pages 121--126. IEEE, 2021.

\bibitem{hermes}
Anthony Kougkas, Hariharan Devarajan, and Xian-He Sun.
\newblock Hermes: a heterogeneous-aware multi-tiered distributed i/o buffering
  system.
\newblock In {\em Proceedings of the 27th International Symposium on
  High-Performance Parallel and Distributed Computing}, pages 219--230, 2018.

\bibitem{adaptive}
Chandra Krintz and Sezgin Sucu.
\newblock Adaptive on-the-fly compression.
\newblock {\em IEEE Transactions on Parallel and Distributed Systems},
  17(1):15--24, 2006.

\bibitem{quiver}
Abhishek~Vijaya Kumar and Muthian Sivathanu.
\newblock Quiver: An informed storage cache for deep learning.
\newblock In {\em 18th $\{$USENIX$\}$ Conference on File and Storage
  Technologies ($\{$FAST$\}$ 20)}, pages 283--296, 2020.

\bibitem{lasch2022cost}
Robert Lasch, Thomas Legler, Norman May, Bernhard Scheirle, and Kai-Uwe
  Sattler.
\newblock Cost modelling for optimal data placement in heterogeneous main
  memory.
\newblock {\em Proceedings of the VLDB Endowment}, 15(11):2867--2880, 2022.

\bibitem{lasch2021workload}
Robert Lasch, Robert Schulze, Thomas Legler, and Kai-Uwe Sattler.
\newblock Workload-driven placement of column-store data structures on dram and
  nvm.
\newblock In {\em Proceedings of the 17th International Workshop on Data
  Management on New Hardware (DaMoN 2021)}, pages 1--8, 2021.

\bibitem{EWAH}
Daniel Lemire, Owen Kaser, and Kamel Aouiche.
\newblock Sorting improves word-aligned bitmap indexes.
\newblock {\em Data \& Knowledge Engineering}, 69(1):3--28, 2010.

\bibitem{Roaring}
Daniel Lemire, Owen Kaser, Nathan Kurz, Luca Deri, Chris O'Hara, Fran{\c{c}}ois
  Saint-Jacques, and Gregory Ssi-Yan-Kai.
\newblock Roaring bitmaps: Implementation of an optimized software library.
\newblock {\em Software: Practice and Experience}, 48(4):867--895, 2018.

\bibitem{liu2019transfer}
Mingyu Liu, Li~Pan, and Shijun Liu.
\newblock To transfer or not: An online cost optimization algorithm for using
  two-tier storage-as-a-service clouds.
\newblock {\em IEEE Access}, 7:94263--94275, 2019.

\bibitem{liu2021keep}
Mingyu Liu, Li~Pan, and Shijun Liu.
\newblock Keep hot or go cold: A randomized online migration algorithm for cost
  optimization in staas clouds.
\newblock {\em IEEE Transactions on Network and Service Management},
  18(4):4563--4575, 2021.

\bibitem{lu2017adaptdb}
Yi~Lu, Anil Shanbhag, Alekh Jindal, and Samuel Madden.
\newblock Adaptdb: Adaptive partitioning for distributed joins.
\newblock {\em Proceedings of the VLDB Endowment}, 10(5), 2017.

\bibitem{mansouricostoptim}
Yaser Mansouri and Abdelkarim Erradi.
\newblock Cost optimization algorithms for hot and cool tiers cloud storage
  services.
\newblock In {\em 2018 IEEE 11th International Conference on Cloud Computing
  (CLOUD)}, pages 622--629, 2018.

\bibitem{mansouridatacenter}
Yaser Mansouri, Adel~Nadjaran Toosi, and Rajkumar Buyya.
\newblock Cost optimization for dynamic replication and migration of data in
  cloud data centers.
\newblock {\em IEEE Transactions on Cloud Computing}, 7(3):705--718, 2019.

\bibitem{shanbhag2017robust}
Anil Shanbhag, Alekh Jindal, Samuel Madden, Jorge Quiane, and Aaron~J Elmore.
\newblock A robust partitioning scheme for ad-hoc query workloads.
\newblock In {\em Proceedings of the 2017 Symposium on Cloud Computing}, pages
  229--241, 2017.

\bibitem{si2022cost}
Wen Si, Li~Pan, and Shijun Liu.
\newblock A cost-driven online auto-scaling algorithm for web applications in
  cloud environments.
\newblock {\em Knowledge-Based Systems}, 244:108523, 2022.

\bibitem{jaya}
Muthian Sivathanu, Midhul Vuppalapati, Bhargav~S Gulavani, Kaushik Rajan, Jyoti
  Leeka, Jayashree Mohan, and Piyus Kedia.
\newblock Instalytics: Cluster filesystem co-design for big-data analytics.
\newblock In {\em 17th $\{$USENIX$\}$ Conference on File and Storage
  Technologies ($\{$FAST$\}$ 19)}, pages 235--248, 2019.

\bibitem{sun2014fine}
Liwen Sun, Michael~J Franklin, Sanjay Krishnan, and Reynold~S Xin.
\newblock Fine-grained partitioning for aggressive data skipping.
\newblock In {\em Proceedings of the 2014 ACM SIGMOD international conference
  on Management of data}, pages 1115--1126, 2014.

\bibitem{liwen_columnar_partitioning}
Liwen Sun, Michael~J. Franklin, Jiannan Wang, and Eugene Wu.
\newblock Skipping-oriented partitioning for columnar layouts.
\newblock {\em Proc. VLDB Endow.}, 10(4):421–432, nov 2016.

\bibitem{vogel2020mosaic}
Lukas Vogel, Viktor Leis, Alexander van Renen, Thomas Neumann, Satoshi Imamura,
  and Alfons Kemper.
\newblock Mosaic: a budget-conscious storage engine for relational database
  systems.
\newblock {\em Proceedings of the VLDB Endowment}, 13(12):2662--2675, 2020.

\bibitem{diesel}
Lipeng Wang, Songgao Ye, Baichen Yang, Youyou Lu, Hequan Zhang, Shengen Yan,
  and Qiong Luo.
\newblock Diesel: A dataset-based distributed storage and caching system for
  large-scale deep learning training.
\newblock In {\em 49th International Conference on Parallel Processing-ICPP},
  pages 1--11, 2020.

\bibitem{wu2002compressing}
Kesheng Wu, Ekow~J Otoo, and Arie Shoshani.
\newblock Compressing bitmap indexes for faster search operations.
\newblock In {\em Proceedings 14th international conference on scientific and
  statistical database management}, pages 99--108. IEEE, 2002.

\bibitem{wu2006optimizing}
Kesheng Wu, Ekow~J Otoo, and Arie Shoshani.
\newblock Optimizing bitmap indices with efficient compression.
\newblock {\em ACM Transactions on Database Systems (TODS)}, 31(1):1--38, 2006.

\bibitem{orpheus}
Liqi Xu, Silu Huang, SiLi Hui, Aaron~J Elmore, and Aditya Parameswaran.
\newblock Orpheusdb: A lightweight approach to relational dataset versioning.
\newblock In {\em Proceedings of the 2017 ACM International Conference on
  Management of Data}, pages 1655--1658, 2017.

\bibitem{yang2021flexpushdowndb}
Yifei Yang, Matt Youill, Matthew Woicik, Yizhou Liu, Xiangyao Yu, Marco
  Serafini, Ashraf Aboulnaga, and Michael Stonebraker.
\newblock Flexpushdowndb: Hybrid pushdown and caching in a cloud dbms.
\newblock {\em Proceedings of the VLDB Endowment}, 14(11):2101--2113, 2021.

\bibitem{yu2020pushdowndb}
Xiangyao Yu, Matt Youill, Matthew Woicik, Abdurrahman Ghanem, Marco Serafini,
  Ashraf Aboulnaga, and Michael Stonebraker.
\newblock Pushdowndb: Accelerating a dbms using s3 computation.
\newblock In {\em 2020 IEEE 36th International Conference on Data Engineering
  (ICDE)}, pages 1802--1805. IEEE, 2020.

\bibitem{compressdb}
Feng Zhang, Weitao Wan, Chenyang Zhang, Jidong Zhai, Yunpeng Chai, Haixiang Li,
  and Xiaoyong Du.
\newblock Compressdb: Enabling efficient compressed data direct processing for
  various databases.
\newblock In {\em Proceedings of the 2022 International Conference on
  Management of Data}, pages 1655--1669, 2022.

\bibitem{zhou2010incorporating}
Jingren Zhou, Per-Ake Larson, and Ronnie Chaiken.
\newblock Incorporating partitioning and parallel plans into the scope
  optimizer.
\newblock In {\em 2010 IEEE 26th International Conference on Data Engineering
  (ICDE 2010)}, pages 1060--1071. IEEE, 2010.

\end{thebibliography}



\end{document}